\documentclass[a4paper,11pt]{article}
\usepackage{latexsym}
\usepackage{amsmath, amsthm, amssymb, bbm}
\usepackage{amsmath, amssymb, bbm}
\usepackage[mathscr]{eucal}
\usepackage{cases}
\usepackage{cite}
\usepackage{hyperref}

\theoremstyle{plain}
\newtheorem{proposition}{Proposition}[section]
\newtheorem{theorem}[proposition]{Theorem}

\theoremstyle{definition}
\newtheorem{definition}[proposition]{Definition}

\theoremstyle{remark}
\newtheorem{remark}[proposition]{Remark}

\newcommand{\rhs}{r.h.s.\ }

\newcommand{\wrt}{w.r.t.\ }
\newcommand{\cf}{cf.\ }

\newcommand{\ud}{\mathrm{d}}
\newcommand{\del}{\partial}

\DeclareMathOperator{\supp}{supp}

\DeclareMathOperator{\Sym}{Sym}

\newcommand{\betrag}[1]{{\lvert #1 \rvert}}

\newcommand{\R}{\mathbb{R}}
\newcommand{\C}{\mathbb{C}}
\newcommand{\Z}{\mathbb{Z}}

\newcommand{\E}{\mathfrak{E}}
\newcommand{\F}{\mathfrak{F}}
\newcommand{\A}{\mathfrak{A}}

\newcommand{\D}{\mathfrak{D}}
\newcommand{\skal}[2]{\langle #1 , #2 \rangle}
\newcommand{\order}{\mathcal{O}}
\newcommand{\id}{\mathrm{id}}
\newcommand{\1}{\mathbbm{1}}

\newcommand{\eps}{\varepsilon}

\newcommand{\diag}{\mathcal{D}}
\newcommand{\Spin}{\mathrm{Spin}}
\newcommand{\Lor}{\mathrm{Lor}}

\newcommand{\Pei}[2]{\lfloor #1, #2 \rfloor}

\newcommand{\WDp}[1]{\colon \negthickspace #1 \! \colon \negthickspace }

\DeclareMathOperator{\WF}{WF}
\DeclareMathOperator{\Had}{Had}
\DeclareMathOperator{\tr}{tr}

\newcommand{\N}{\mathbb{N}}

\newcommand{\reg}{{\mathrm{reg}}}
\newcommand{\loc}{{\mathrm{loc}}}
\newcommand{\ret}{{\mathrm{ret}}}
\newcommand{\adv}{{\mathrm{adv}}}

\newcommand{\CatVec}{\mathbf{Vec}}
\newcommand{\CatAlg}{\mathbf{Alg}}
\newcommand{\CatMan}{\mathbf{Man}}
\newcommand{\CatSpMan}{\mathbf{SpMan}}
\newcommand{\CatGSpMan}{\mathbf{GSpMan}}

\newcommand{\g}{\mathfrak{g}}
\newcommand{\T}{\mathfrak{T}}
\newcommand{\MT}{\mathfrak{MT}}
\newcommand{\gt}{{\hat \otimes}}
\newcommand{\Cl}{\mathrm{Cl}}
\DeclareMathOperator{\Mat}{Mat}
\DeclareMathOperator{\ad}{ad}
\newcommand{\Lie}{\mathcal{L}}

\begin{document}

\begin{flushright}
UWThPh-2012-32
\end{flushright}
\hfill
\vspace{0.7cm}
\hfill
\begin{center}
{\LARGE The renormalized locally covariant Dirac field} \\
\hfill
\vspace{0.5cm}
\hfill \\
{\large Jochen Zahn \\  \vspace{0.3cm} Fakult\"at f\"ur Physik, Universit\"at Wien, \\ Boltzmanngasse 5, 1090 Wien, Austria \\ jochen.zahn@univie.ac.at \\
\hfill \\
\today \\}
\hfill \\
\end{center}

%\markboth{Jochen Zahn}
%{The renormalized locally covariant Dirac field}
%
%\title{The renormalized locally covariant Dirac field}
%%\author{Jochen Zahn \\  \\ Fakult\"at f\"ur Physik, Universit\"at Wien, \\ Boltzmanngasse 5, 1090 Wien, Austria. \\ jochen.zahn@univie.ac.at}
%
%\author{Jochen Zahn}
%\address{Fakult\"at f\"ur Physik, Universit\"at Wien, Boltzmanngasse 5, 1090 Wien, Austria.\\
%\email{jochen.zahn@univie.ac.at}}
%
%%\date{October 15, 2012}
%
%\maketitle
%
%\begin{history}
%\received{(Day Month Year)}
%\revised{(Day Month Year)}
%\end{history}

\begin{abstract}
The definition of the locally covariant Dirac field is adapted such that it may be charged under a gauge group and in the presence of generic gauge and Yukawa background fields. We construct renormalized Wick powers and time-ordered products. It is shown that the Wick powers may be defined such that the current and the stress-energy tensor are conserved, and the remaining ambiguity is characterized. We sketch a variant of the background field method that can be used to determine the renormalization group flow at the one loop level from the nontrivial scaling of Wick powers.
\end{abstract}

%\keywords{Quantum field theory over curved backgrounds; spinor fields.}

%\ccode{Mathematics Subject Classification 2010: 81T20, 81T05}

\section{Introduction}

The last one and a half decades saw an impressive revival of the theory of quantum fields on curved spacetimes. This was initiated by Radzikowski's discovery that Hadamard two-point functions can be equivalently characterized in terms of their wave front set \cite{Radzikowski}. This lead Brunetti, Fredenhagen and K\"ohler to the formulation of the \emph{microlocal spectrum condition} and the construction of Wick polynomials \cite{BrunettiFredenhagenKoehler}. Using a local renormalization scheme \`a la Epstein and Glaser and Steinmann's concept of the scaling degree, Brunetti and Fredenhagen were able to prove the perturbative renormalizability of the $\varphi^4$ model on generic spacetimes \cite{BrunettiFredenhagenScalingDegree}. What was missing was some means to compare field theories defined on different spacetimes, or, put differently, to define one theory coherently on all spacetimes. This was provided by the \emph{generally covariant locality principle} introduced by Brunetti, Fredenhagen and Verch \cite{BrunettiFredenhagenVerch}. This principle is naturally formulated in categorical language: One starts with the category $\CatMan$ of globally hyperbolic manifolds, with causal isometric embeddings as morphisms.\footnote{A categorical language was already used by Dimock \cite{DimockDirac}, who, however, used isometries as morphisms, instead of isometric embeddings. Hence, the crucial requirement of a local construction of algebras and fields is missing in that framework.} A locally covariant theory is then a functor from $\CatMan$ to the category of ($C^*$)-algebras with injective homomorphisms as morphisms. The concept of a locally covariant theory was essential for the definition of covariant Wick powers and time-ordered products due to Hollands and Wald \cite{HollandsWaldWick, HollandsWaldTO}.

The framework was also crucial for the proof of the spin-statistics theorem on curved backgrounds \cite{VerchSpinStatistics}. Examples of further applications are the discussion of quantum energy inequalities \cite{FewsterPfenningQEILocCovI} and the renormalization group in curved spacetimes \cite{HollandsWaldRG}. The framework was also used in the treatments of Yang--Mills gauge fields \cite{HollandsYM}, perturbative (classical) gravity \cite{RejznerFredenhagen}, and the quantization of submanifold embeddings \cite{BRZ}.

The locally covariant Dirac field was first considered by Verch \cite{VerchSpinStatistics}, and later worked out by Sanders \cite{SandersDirac}. The crucial point is the replacement of the category $\CatMan$ by the category $\CatSpMan$, which also captures the spin structure. However, the spacetime and the spin structure were the only allowed non-trivial backgrounds. Furthermore, only linear fields were incorporated, i.e., no Wick powers and time-ordered products. The latter problem was treated by Dappiaggi, Hack and Pinamonti, who provided a definition of Wick powers in order to be able to discuss backreaction effects through the semiclassical Einstein equation \cite{DHP09}. But their proposal has some shortcomings, to be commented on below.

The aim of the present paper is to generalize and extend the framework of Sanders. The generalization consists in allowing for non-trivial gauge and Yukawa background fields. This is achieved by further extending the underlying category $\CatSpMan$ to the category $\CatGSpMan$, which also includes the principal bundle corresponding to the gauge group, a gauge potential, and a scalar field (describing the Yukawa background). In particular, gauge transformations then correspond to morphisms of the category. 

We extend Sanders' work in that we also treat non-linear fields (Wick powers) and interactions (through time-ordered products). Building on the work of Rejzner \cite{RejznerFermions} on fermionic fields on Minkowski space, we work in the framework of perturbative algebraic quantum field theory (pAQFT) \cite{BDF09}, i.e., by deformation quantization of a graded commutative algebra of functionals.
The first step is to define the algebra of so-called \emph{microcausal} functionals. The crucial point is to show that Hadamard two-point functions exist, a result that is a rather straightforward generalization of results of Fewster and Verch \cite{FewsterVerchDirac}.
The next step is to define Wick powers. This is done via Hadamard parametrices, and the first task is to define what a covariant choice of a parametrix actually is. The next is then to show that parametrices exist.
Our treatment requires less assumptions than the existing ones \cite{SahlmannVerchHadamard, DHP09}, in that we allow for a coupling to non-trivial gauge and Yukawa backgrounds. 
Finally, we present a construction of time-ordered products, by a generalization of the work of Hollands and Wald for the scalar case \cite{HollandsWaldTO}.

We also provide some applications of the framework. We show that a conserved current operator can always be achieved and discuss the remaining renormalization freedom. This local and covariant definition of the current could also be useful for the study of backreaction effects in quantum electrodynamics on Minkowski space in the presence of an electromagnetic background field.
Furthermore, we show that, provided the nontrivial background consists only of gravity and a constant mass, there is no algebraic obstruction to achieving a conserved stress-energy tensor, for any spacetime dimension.
We also classify the remaining ambiguities, thereby proving a conjecture of \cite{DHP09}. As another application, we sketch the determination of the renormalization group flow, at first order in $\hbar$, via a kind of background field method, solely on the basis of the scaling behavior of the parametrix, i.e., without calculating any loop integral.

The article is structured as follows: In the next section, we introduce the categorical setup, which now also includes a principal $G$-bundle and a background gauge connection and Yukawa field. We also introduce the classical algebra of functionals. In Section~\ref{sec:Quantization}, the quantization of the algebra of functionals, via deformation quantization, is described. Also the construction of covariant Wick powers and time-ordered products is performed. The applications to current and stress-energy conservation and the renormalization group flow are described in Section~\ref{sec:Currents}. In \ref{app:Spin}, we recall some basic notions of spin geometry and in \ref{app:Deformation}, we provide a proof of a proposition on deformations of spacetimes and associated structures.

\subsection{Notation and Conventions}

We are working on $n$-dimensional spacetimes with signature $(-,+, \dots, +)$. For morphisms and equivalences of principal bundles, we use the following definition:
\begin{definition}
A \emph{morphism} $\eta$ between two principal $G$ bundles $P, P'$ over manifolds $M$ and $M'$ is a smooth map $\eta : P \to P'$ which is $G$-equivariant, i.e., $\eta(p g) = \eta(p) g$ for all $p \in P$, $g \in G$, and covers a smooth map $\chi: M \to M'$, i.e., $\pi' \circ \eta = \chi \circ \pi$. $P$ and $P'$ are \emph{equivalent}, $P \simeq P'$, if $\eta$ and $\chi$ are diffeomorphisms.
\end{definition}

The Cartesian product of bundles $E, F$ is denoted by $E \boxtimes F$, which is a bundle over the Cartesian product of the base spaces. Smooth sections of a bundle $E$ with base space $M$ are denoted by $\Gamma^\infty(M, E)$, and a subscript $c$ denotes compactly supported sections. $\dot T^* M$ denotes the cotangent bundle of $M$, with the zero removed. For a manifold $M$, $\diag^k \subset M^k$ denotes the total diagonal,
\[
 \diag^k = \{(x, \dots, x) \in M^k\}.
\]

For a half-integer $k$, $[k]$ denotes the integer part. 
The symbol $\doteq$ denotes a definition of the left hand side by the right hand side. Typically, primed symbols, such as $v'$, stand for elements of a dual space (an exception is a primed coordinate $x'$).

\section{The categorical description}

Before introducing the coupling to background fields, let us first review the structure introduced in \cite{SandersDirac}. The identity component of the $\Spin$ group is denoted by $\Spin_0$, \cf \ref{app:Spin} for a definition. A \emph{spin structure} $SM$ over an oriented, time-oriented spacetime $M$ is a principal $\Spin_0$ bundle over $M$ with a projection $\pi_S: SM \to FM$ to the bundle of oriented, time-oriented, orthonormal frames, which preserves the base point and intertwines the action of $\Spin_0$, i.e.,
\[
 \pi_S \circ S = \lambda(S) \circ \pi_S,
\]
where $S \in \Spin_0$ and $\lambda$ is the covering map to the connected component $\Lor_0$ of the Lorentz group. One defines the following category:
\begin{description}
\item[$\CatSpMan$:] The objects are spin structures $SM$ whose base spaces $M$ are oriented, time-oriented, globally hyperbolic spacetimes. A morphism $\chi: SM \to SM'$ is a principal $\Spin_0$ bundle morphism, covering an orientation, time-orientation and causality preserving isometric embedding $\psi: M \to M'$ such that $\pi'_S \circ \chi = \psi_* \circ \pi_S$.
\end{description}

In order to be able to functorially associate vector spaces and algebras to such spin structures, we also introduce the following categories.

\begin{description}
\item[$\CatVec_{(i)}$:] The objects are locally convex vector spaces. The morphisms are continuous linear (injective) maps.
\item[$\CatAlg$:] The objects are topological $*$-algebras. The morphisms are continuous injective $*$-algebra homomorphisms.
\end{description}

As discussed in \ref{app:Spin}, there is a standard (spinor) representation $\rho_0$ of $\Spin_0$ on $\C^{2^{[n/2]}}$.
The associated vector bundle $DM$ induced by this representation is called the \emph{standard} Dirac bundle in \cite{SandersDirac}. Its dual bundle is denoted by $D^*M$.
We note that there are anti-linear conjugations\footnote{Of course there is also a charge conjugation. As our aim is to study arbitrary background gauge fields, where charge conjugation is not a symmetry, we do not discuss it here.}
\begin{align*}
 ^+ : DM & \to D^*M, & ^+: D^*M & \to DM,
\end{align*}
fulfilling the usual properties, defined through
\begin{align*}
 [p, z]^+ & \doteq [p, z^+], & [p, z']^+ & \doteq [p, z'^+],
\end{align*}
where $p \in P$ and
\begin{align*}
 z^+ & \doteq - i z^* \gamma^0 , & z'^+ & \doteq - i \gamma^0 z'^*,
\end{align*}
for $z \in \C^{2^{[n/2]}}, z' \in {\C^{2^{[n/2]}}}^*$, \cf \ref{app:Spin} and \cite{SandersDirac, HackThesis} for details on the case $n=4$.

We now want the Dirac field to be charged under a compact Lie group $G$ in a representation $\rho$. Hence, we consider a principal $G$ bundle $P$ over $M$, and consider the direct product bundle\footnote{We refer to \cite[p.~82]{KobayashiNomizu} for a definition.} $SM + P$.\footnote{Given the fact that we also might want to consider scalar fields charged under $G$, it seems reasonable to consider only direct products of $SM$ and $P$ and not general principal $\Spin_0 \times G$ bundles.} On $P$, we consider a \emph{connection}, i.e., a $\g$ valued 1-form $A$ on $P$, which is equivariant and fulfills $A(v^\#) = v$, where $v^\#$ is the fundamental vector field corresponding to $v \in \g$, \cf \cite[Chapter II]{KobayashiNomizu}. We recall that the Levi-Civita connection induces a unique \emph{spin connection} $\Omega$ on $SM$, \cf \cite[Section~II.4]{LawsonMichelsohn} for details. By \cite[Prop.~II.6.3]{KobayashiNomizu}, there is then a unique connection on $SM + P$ such that the pushforward under the projection homomorphisms coincide with $\Omega$ and $A$. We also want to allow for couplings to a nonconstant Yukawa background field $m \in C^\infty(M, \R)$. This leads us to consider the following category:
\begin{description}
\item[$\CatGSpMan$:] The objects are quadruples $(SM, P, A, m)$, where $SM$ is a spin structure over an oriented, time-oriented globally hyperbolic spacetime $M$, $P$ a principal $G$ bundle over $M$, $A$ a connection on $P$, and $m \in C^\infty(M, \R)$. A morphism $\chi: (SM, P, A, m) \to (SM', P', A', m')$ is given by $(\chi_{SM}, \chi_P)$, where $\chi_{SM (P)}$ is a principal $\Spin_0$ ($G$) bundle morphism. $\chi_{SM}$ and $\chi_G$ cover the same orientation, time-orientation and causality preserving isometric embedding $\psi: M \to M'$ with $m = \psi^* m'$. Furthermore,
%$\pi' \circ \chi = \psi_* \circ \pi$ and
$A = \chi_P^* A'$.
\end{description}

We note that a pair $(\chi_{SM}, \chi_P)$ as above induces a principal $\Spin_0 \times G$ bundle morphism $\chi: SM + P \to SM' + P'$ by $\chi (p, q) = (\chi_{SM} (p), \chi_P (q))$, where $p \in SM|_x$, $q \in P|_x$ for some $x \in M$. We also remark that taking $SM' = SM$, $P = P'$, $\chi_{SM} = \id$, and, in a local trivialization,
\begin{equation}
\label{eq:GaugeTrafo}
 \chi_P:(x,g) \mapsto (x, h(x) g)
\end{equation}
for some $h \in C^\infty(M, G)$ corresponds to a gauge transformation. Hence, gauge equivalence is built into the categorical framework.

\begin{remark}
The incorporation of background fields other than the gravitational one into the framework of locally covariant field theory can be found in earlier works, for example \cite{BaerGinoux} (implicitly through the specification of a Green-hyperbolic operator) or \cite{HollandsWaldWick} (though not formalized in the language of category theory). A unified treatment of gauge and general covariance can be found in \cite{HollandsYM} (again not in the language of category theory). But, as explained below in Remark~\ref{rem:GaugeCovariance}, our approach has a different notion of local covariance. 
\end{remark}

Given a representation $\rho$ of $G$ on a finite dimensional $\C$ vector space $V$, we construct the vector bundle $D_\rho M$ associated to $SM + P$ via the representation $\rho_0 \otimes \rho$ on $\C^{2^{[n/2]}} \otimes V$. The corresponding dual bundle is denoted by $D_\rho^* M$, and the double spinor bundle by $D_\rho^\oplus M \doteq D_\rho M \oplus D_\rho^* M$. We define the vector spaces
\begin{align*}
 \E^{(*)}(SM, P) & \doteq \Gamma^\infty(M, D_\rho^{(*)} M), \\
 \E^\oplus(SM, P) & \doteq \Gamma^\infty(M, D_\rho^\oplus M).
\end{align*}
The assignments $(SM, P, A, m) \mapsto \E^{(*)}(SM, P), \E^\oplus(SM, P)$ are contravariant functors from $\CatGSpMan$ to $\CatVec$. Under $\E$, the morphism $\chi$ is mapped to the pullback $\xi^*$ of $\xi: D_\rho M \to D_\rho M'$, defined by $\xi([p, z]) = [\chi(p), z]$, $p \in SM + P$, $z \in \C^{2^{[n/2]}} \otimes V$, and analogously for $\E^*$, $\E^\oplus$.
Note that the pull-back $\xi^*$ is well-defined, as $\xi$ reduces to an isomorphism of fibers.
We also define the test section spaces
\begin{align*}
 \D^{(*)}(SM, P) & \doteq \Gamma^\infty_c(M, D_\rho^{(*)} M), \\
 \D^\oplus(SM, P) & \doteq \Gamma^\infty_c(M, D_\rho^\oplus M).
\end{align*}
These are covariant functors from $\CatGSpMan$ to $\CatVec_i$. A morphism $\chi$ is mapped to the push-forward $\xi_*$, where $\xi$ is defined as above and $\xi_*$ is extended from $\chi(M)$ to $M'$ by the zero section. For later convenience, we also introduce
\begin{equation}
\label{eq:Tens}
 \T_c(SM, P) \doteq \Gamma^\infty_c(M, \wedge( D_\rho^\oplus M \otimes T^\oplus M)),
\end{equation}
where\footnote{Here and in the following, $\bigoplus$ denotes a finite direct sum, i.e., only a finite number of components is nonzero.}
\[
 T^\oplus M \doteq \bigoplus_{k} \Sym^k TM,
\]
$\wedge$ denotes the exterior tensor product, and $\Sym^k$ the $k$th symmetric tensor product. This is a also a covariant functor from $\CatGSpMan$ to $\CatVec_i$. Sometimes we need to be more specific, then $\T_c^{j A}$ denotes the subspace where the $j$th exterior power is taken, and $A \in \N_0^j$ counts the tensor power corresponding to $T^\oplus M$ in each of the factors. For example, $\T_c^{1 0} = \D^\oplus$.

As $V$ is finite dimensional, $V \simeq \C^N$, there is an inner product on $V$. By averaging over $G$, we obtain a sesquilinear form $\skal{\cdot}{\cdot}_V$ on $V$ that is conserved under the action $\rho$. There is then a natural anti-linear map $^+$ from $V$ to $V^*$, given by
\[
 v^+(w) \doteq \skal{v}{w}_V.
\]
Analogously, we may define $^+: V^* \to V$. Thus, we may define the conjugation map $^+: D_\rho M \to D_\rho^* M$ by
\[
 [p, z \otimes v]^+ \doteq [p, z^+ \otimes v^+], \quad p \in SM + P, z \in \C^{2^{[n/2]}}, v \in V,
\]
and analogously for $^+: D_\rho^{*} M \to D_\rho M$. This lifts to anti-linear maps $\E(SM, P) \to \E^*(SM, P)$, $\E^*(SM, P) \to \E(SM, P)$, and hence to an anti-linear map $^+: \E^\oplus(SM, P) \to \E^\oplus(SM, P)$.
The pointwise pairing $D_\rho^* M|_x \times D_\rho M|_x \to \C$ defined by
\begin{multline*}
 \skal{[p, z' \otimes v']}{[p,z \otimes v]} \doteq z'(z) v'(v), \\ p \in SM + P, z \in \C^{2^{[n/2]}}, v \in V, z' \in {\C^{2^{[n/2]}}}^*, v' \in V^*,
\end{multline*}
leads to a pairing $\E^*(SM, P) \times \E(SM, P) \to C^\infty(M)$, and to a pairing $\E^\oplus(SM, P) \times \E^\oplus(SM, P) \to C^\infty(M)$ defined
by\footnote{Note that we are using a different convention than in \cite{SandersDirac} and \cite{FewsterVerchDirac}, in that we are contracting the spinor with the cospinor and vice versa.}
\begin{equation}
\label{eq:Pairing}
 \skal{(f, f')}{(g, g')} \doteq \skal{g'}{f} + \skal{f'}{g}, \quad f, g \in \E(SM, P), f', g' \in \E^*(SM, P).
\end{equation}

\subsection{The Dirac operator and its fundamental solutions}
\label{sec:DiracOperator}

The connection on $SM + P$ induces the \emph{exterior covariant derivative} $\ud_A$ on $D_\rho M$, \cf \cite[Sec.~II.5]{KobayashiNomizu}. This determines a covariant derivative $\nabla$ on sections of $D_\rho M$, \cite[Sec.~III.1]{KobayashiNomizu}. Analogously, there is a covariant derivative $\nabla^*$ on sections of $D_\rho^* M$. We may then define the Dirac operators $D$ and $D^*$, which, in a local trivialization\footnote{Here we use a trivialization of tensor product form, i.e., the trivialization of $D_\rho M$ is induced from trivializations of $SM \times_{\rho_0} \C^{2^{[n/2]}}$ and $P \times_\rho V$.} are given by\footnote{Here we use the customary notation $\gamma^\mu = \gamma(\ud x^\mu)$ for the Clifford multiplication composed with the spinor representation.}
\begin{align}
\label{eq:D}
 D & = - \gamma^\mu (\del_\mu + \Omega_\mu - i A_\mu) + m = - \gamma^\mu \nabla_\mu + m, \\
 D^* & = \gamma^\mu (\del_\mu + \Omega_\mu^* + i A_\mu^*) + m = \gamma^\mu \nabla^*_\mu + m, \nonumber
\end{align}
where $\Omega_\mu$ is the spin connection coefficient, $m$ is the smooth function in the objects of $\CatGSpMan$, and $A_\mu$ is determined from the connection $A$ in the objects of $\CatGSpMan$ by pull-back \wrt the local section defining the trivialization. The $*$ on $\Omega_\mu$ and $A_\mu$ denotes the action on the dual bundle, defined by duality.
These operators intertwine the action of $\E^{(*)}(\chi)$ and $\D^{(*)}(\chi)$ for a morphism $\chi$ of $\CatGSpMan$, i.e., they are natural transformations $\E^{(*)} \to \E^{(*)}$, $\D^{(*)} \to \D^{(*)}$.

Let us briefly review the construction of retarded and advanced propagators and fundamental solutions for $D$ and $D^*$. The square of $D$ is a normally hyperbolic operator \cite{BGP07},
\begin{equation}
\label{eq:P}
 P = D D = g^{\mu \nu} \nabla_\mu \nabla_\nu - 2 m \gamma^\mu \nabla_\mu + \tfrac{1}{4} [\gamma^\mu, \gamma^\nu] ( \mathfrak{R}_{\mu \nu} - i F_{\mu \nu} ) - \gamma^\mu \del_\mu m + m^2,
\end{equation}
where
\[
 [\nabla_\mu,  \nabla_\nu] = \mathfrak{R}_{\mu \nu} - i F_{\mu \nu},
\]
with $F$ the curvature of the connection $A$ in the representation $\rho$ and $\mathfrak{R}$ the spin curvature \cite[Section~II.4]{LawsonMichelsohn}.
To the normally hyperbolic operator $P$ correspond unique retarded and advanced propagators \cite{BGP07}
\[
 \Delta_{\ret / \adv} : \D(SM, P) \to \E(SM, P).
\]
The corresponding propagators for $D$ are then defined as
\begin{equation}
\label{eq:DefS_ra}
 S_{\ret / \adv} \doteq D \circ \Delta_{\ret / \adv}.
\end{equation}
For $D^*$, one proceeds in complete analogy, arriving at propagators $S_{\ret / \adv}^*$.
By construction, one then has $D^{(*)} \circ S_{\ret / \adv}^{(*)} = \id$. A theorem by Dimock \cite{DimockDirac} (see also \cite{Muehlhoff}), which is straightforwardly generalized to fields charged under a gauge group, implies that then also $S_{\ret / \adv}^{(*)} \circ D^{(*)} = \id$ on $\D^{(*)}(SM, P)$. Hence, $S_{\ret / \adv}^{(*)}$ are the unique retarded/advanced propagators for $D^{(*)}$, and the causal propagator is given by $S^{(*)} = S_\ret^{(*)} - S_\adv^{(*)}$. For the double spinor notation, we define\footnote{This is the Dirac operator obtained from the variation of the Dirac action, \cf \cite{RejznerFermions}. In other works \cite{DHP09, SandersDirac}, the double spinor Dirac operator is defined as $D \oplus D^*$.}
\begin{align*}
 D^\oplus & \doteq D \oplus - D^*, & S^\oplus & \doteq S \oplus - S^*.
\end{align*}
As the $S_{\ret / \adv}^{(*)}$ are unique and $D^{(*)}$ is a natural transformation,
%intertwines $\E^{(*)}$ and $\D^{(*)}$,
we have, for a morphism $\chi: (SM, P, A, m) \to (SM', P', A', m')$,
\[
 S^{(*)} = \E^{(*)}(\chi) \circ {S'}^{(*)} \circ \D^{(*)}(\chi).
\]
We also note that $S^{(*)}$ fulfills
\begin{equation}
\label{eq:S}
 \int \skal{f'}{S f}(x) \ud_g x = - \int \skal{S^* f'}{f}(x) \ud_g x = - \int \overline{\skal{f^+}{S f'^+}(x)} \ud_g x,
\end{equation}
where $f \in \D(SM, P)$, $f' \in \D^*(SM, P)$, and $\ud_g x$ is the canonical volume form. The first equality can be shown as in the proof of Theorem~2.1 in \cite{DimockDirac}. The second equality follows from $(D f)^+ = D^* f^+$ and the uniqueness of the retarded/advanced propagators. Finally, we remark that $S^\oplus$ may also be seen as a distribution, $S^\oplus \in \Gamma^\infty_c(M^2, D_\rho^\oplus M \boxtimes D_\rho^\oplus M)'$, by
\[
 S^\oplus(u,v) \doteq \int \skal{u}{S^\oplus v}(x) \ud_g x,
\]
where we used the pairing \eqref{eq:Pairing}. Analogously, $S \in \Gamma^\infty_c(M^2, D_\rho^* M \boxtimes D_\rho M)'$, and
\[
 S^\oplus((f,f'),(g,g')) = S(f', g) + S(g', f),
\]
where we used \eqref{eq:S}.

\subsection{Functionals}

In the framework of pAQFT, one considers the algebra of functionals on the configuration space and deforms it (quantization). For fermionic fields, it was proposed in \cite{RejznerFermions} to consider functionals on the space of antisymmetrized configurations, i.e., in the present setting, on
\[
 \wedge \E^\oplus(SM, P) \doteq  \bigoplus_{k = 0}^\infty \wedge^k \E^\oplus(SM, P),
\]
with
\[
 \wedge^k \E^\oplus(SM, P) \doteq \{ B \in \Gamma^\infty(M^k, (D_\rho^\oplus M)^k) | B \text{ antisymmetric} \}.
\]
This space is equipped with its natural topology (uniform convergence of all derivatives on compact subsets). For an element $B \in \wedge \E^\oplus(SM, P)$, we denote by $B_k$ its component in $\wedge^k \E^\oplus(SM, P)$. 

We now consider functionals on $\wedge \E^\oplus(SM, P)$, i.e., linear maps from this space into the complex numbers. We denote by $F_k$ the restriction of a functional $F$ to $\wedge^k \E^\oplus(SM, P)$. Then we define the grade by $\betrag{F_k} = k$. The \emph{regular} functionals, $\F_\reg(SM, P)$, are those of the form
\begin{equation}
\label{eq:F_reg}
 F_k(B) = \int \skal{f_k}{B_k}(x_1, \dots, x_k) \ud_g x_1 \dots \ud_g x_k,
\end{equation}
with $f_k \in \Gamma^\infty_c(M^k, D_\rho^\oplus M^k)$, $f_k$ antisymmetric. We call $f_k$ the \emph{kernel} of $F_k$. Here we used the obvious generalization of the pairing \eqref{eq:Pairing}. We can introduce an antisymmetric product $\wedge$ on $\F_\reg(SM, P)$, by defining the kernel of the product $H = F \wedge G$ as
\begin{multline*}
 h_k(x_1, \dots, x_k) \\ \doteq \sum_{l=0}^k \frac{1}{l!(k-l)!}  \sum_{\pi \in S_k} (-1)^{\betrag{\pi}} f_l(x_{\pi(1)}, \dots, x_{\pi(l)}) g_{k-l}(x_{\pi(l+1)}, \dots, x_{\pi(k)}).
\end{multline*}
An involution on $\F_\reg(SM, P)$ is defined as
\[
 F^*(B) \doteq \overline{F(B^+)},
\]
where on elements of $\wedge \E^\oplus(SM, P)$, conjugation is defined by
\[
 (u_1 \wedge \dots \wedge u_k)^+ = u_k^+ \wedge \dots \wedge u_1^+.
\]
Finally, we equip $\F_\reg(SM, P)$ with the topology induced from the standard locally convex topology on $\Gamma^\infty_c(M^k, D_\rho^\oplus M^k)$ (uniform convergence of all derivatives on compact sets), the space of the kernels.
The assignment $(SM, P, A, m) \mapsto \F_\reg(SM, P)$ is then a covariant functor from $\CatGSpMan$ to $\CatAlg$.

The regular functionals do not allow for the description of local interactions or nonlinear observables, such as the stress--energy tensor. In order to cure this, one allows for more general kernels $f_k$, namely compactly supported distributions fulfilling the wave front set condition
\begin{equation*}
 \WF(f_k) \cap (\bar V_+^k \cup \bar V_-^k) = \emptyset,
\end{equation*}
where $\bar V_{\pm}$ is the closure of the dual of the forward/backward light cone. These are called the \emph{microcausal} functionals. They also form an algebra $\F(SM, P)$. It can be equipped with a topology such that it is a nuclear, locally convex vector space \cite{BDF09, Rej11}. $\F$ is then also a covariant functor from $\CatGSpMan$ to $\CatAlg$.

By reference to the support of the kernels $f_k$, one defines the \emph{support} of a functional as
\begin{equation}
\label{eq:suppF}
 \supp F = \left\{ x \in M \mid (x, x_2, \dots, x_k) \in \supp f_k \text{ for some } x_i \right\}. 
\end{equation}
Here we assumed without loss of generality that $f_k$ is antisymmetric.
The subspace $\F_\loc(SM, P)$ of $\F(SM, P)$ in which the $f_k$'s are localized on the total diagonal $\diag^k$
and their wave front sets orthogonal to $T \diag^k$,
\begin{equation*}
 \WF(f_k) \perp T \diag^k,
\end{equation*}
is the space of \emph{local} functionals. It is a covariant functor from $\CatGSpMan$ to $\CatVec_i$.

We denote by $\F_0(SM, P)$ the ideal of functionals that vanish on on-shell configurations, i.e., on configurations fulfilling $D^\oplus B = 0$, where $D^\oplus$ acts on an arbitrary coordinate. We define the on-shell functionals as $\F_S(SM, P) \doteq \F(SM, P) / \F_0(SM, P)$. This amounts to identifying two functionals if they agree on all on-shell configurations. Due to the functoriality of the Dirac operator, this is also a covariant functor from $\CatGSpMan$ to $\CatAlg$.

\section{Quantization}
\label{sec:Quantization}

To prepare grounds for the deformation of the graded commutative algebra $\F$ in the spirit of deformation quantization \cite{DuetschFredenhagenDeformation}, we first have to equip it with a Poisson structure by defining the Peierls bracket.  To this avail, we introduce functional derivatives \cite{RejznerFermions}
\[
 F^{(1)}(B)(u) \doteq F(u \wedge B), \quad B \in \wedge \E^\oplus(SM, P), u \in \E^\oplus(SM, P).
\]
Hence, $F^{(1)}(B)$ can be interpreted as a compactly supported distributional section of $D_\rho^\oplus M$. We denote its integral kernel by $F^{(1)}(B)(x)$. For $F \in \F_\reg$, this is even a smooth section. Higher order derivatives are defined by composition of derivatives, i.e.,
\[
 F^{(k)}(B)(u_1, \dots, u_k) = F(u_1 \wedge \dots \wedge u_k \wedge B).
\]

Given the fundamental solution $S^\oplus$, the Peierls bracket of two observables $F, G \in \F_\reg$, with $F$ being homogeneous, is defined as
\[
 \Pei{F}{G} \doteq (-1)^{\betrag{F}+1} \int F^{(1)}(x) \wedge G^{(1)}(y) S^\oplus(x,y) \ud_g x \ud_g y.
\]
Note that here and in the following, the contraction of $F^{(1)}$ and $G^{(1)}$ with $S^\oplus$ has to be understood as in the pairing defined in \eqref{eq:Pairing}.

In deformation quantization, one aims at finding a product $\star$ on the observables, fulfilling
\begin{align}
\label{eq:deformationQuantization}
 F \star G & = F \wedge G + \order(\hbar), & F \star G - (-1)^{\betrag{F} \betrag{G}} G \star F & = i \hbar \Pei{F}{G} + \order(\hbar^2),
\end{align}
in the sense of formal power series in $\hbar$. This is straightforward for the regular functionals \cite{RejznerFermions}. We define the operator $\Gamma^\otimes_{\frac{i}{2} S}$ by
\[
 \Gamma^\otimes_{\frac{i}{2}S} (F \otimes G) \doteq (-1)^{\betrag{F}+1} \frac{i}{2} \int F^{(1)}(x) \otimes G^{(1)}(y) S^\oplus(x,y) \ud_g x \ud_g y,
\]
and the $\star$ product as
\[
 F \star G \doteq \wedge \exp(\hbar \Gamma^\otimes_{\frac{i}{2} S}) F \otimes G.
\]
Here the wedge denotes the wedge product, $\wedge (F \otimes G) \doteq F \wedge G$. It is clear that \eqref{eq:deformationQuantization} is fulfilled.

As $S^\oplus$ is a bi-solution, $\star$ is also well-defined on the regular on-shell functionals. As the fundamental solution is a local and covariant object, the assignment $(SM,P) \mapsto (\F_\reg(M)[[\hbar]], \star)$ is a covariant functor from $\CatGSpMan$ to $\CatAlg$.

The extension to microcausal functionals proceeds via Hadamard two-point functions. These are defined as follows:
\begin{definition}
A \emph{Hadamard two-point function} is a distributional section $\omega \in \Gamma^\infty_c(M^2, D_\rho^\oplus M^2)'$ fulfilling
\begin{align}
\label{eq:HadamardWaveEq}
 \omega(D^\oplus u, v) & = 0, \\
\label{eq:HadamardAnticommutator}
 \omega(u,v) + \omega(v,u) & = i S^\oplus(u,v), \\
\label{eq:HadamardConjugation}
 \overline{\omega(u, v)} & = \omega(v^+,u^+), \\
\label{eq:HadamardWF}
 \WF(\omega) & \subset C_+,
\end{align}
where $u, v \in \Gamma^\infty_c(M, D_\rho^\oplus M)$ and
\[
 C_\pm = \{ (x_1, x_2; k_1, - k_2) \in T^* M^2 \setminus \{ 0 \} | (x_1; k_1) \sim (x_2; k_2), k_1 \in \bar V^\pm_{x_1} \}.
\]
Here $(x_1; k_1) \sim (x_2; k_2)$ if there is a lightlike geodesic joining $x_1$ and $x_2$ to which $k_1$ and $k_2$ are co-parallel, and $k_2$ coincides with the parallel transport of $k_1$ along this curve. For $x_1 = x_2$, $k_1, k_2$ are lightlike and coinciding.
\end{definition}

Assume for the moment that such distributions exist for all $(SM, P)$ (this is shown later). Denote by $\omega_a(u, u') = \frac{1}{2} ( \omega(u,u') - \omega(u', u))$ the antisymmetric part and define a product $\star_\omega$, equivalent to $\star$,
\begin{equation}
\label{eq:scalarStarOmega}
 F \star_\omega G \doteq \alpha_{\omega_a} \left( \alpha_{\omega_a}^{-1} F \star \alpha_{\omega_a}^{-1} G \right),
\end{equation}
by the equivalence map
\[
  \alpha_{\omega_a} \doteq \exp(\hbar \Gamma_{\omega_a}),
\]
with
\[
  \Gamma_{\omega_a} F \doteq \int \ud_gx \ud_gy \ \omega_a(x,y) F^{(2)}(x,y).
\]
By \eqref{eq:HadamardAnticommutator}, the $\star_\omega$ product amounts to replacing $\frac{i}{2} S^\oplus$ by $\omega$ in the definition of $\star$. The condition \eqref{eq:HadamardConjugation} ensures that $\star_\omega$ is compatible with the conjugation.
From \eqref{eq:HadamardWaveEq} it follows that also $\star_\omega$ is well-defined on on-shell functionals. Furthermore, due to condition \eqref{eq:HadamardWF}, $\star_\omega$ can be extended to the microcausal functionals $\F(SM, P)$, \cf \cite{BDF09} for the scalar case. To achieve a fully covariant construction, it is convenient to consider all possible $\omega$'s at the same time. Hence, we define $\Had(SM, P)$ to be the set of all Hadamard two-point functions. One then defines $\A(SM, P)$ as the space of families
\begin{align*}
 F & = \{ F_\omega \}_{\omega \in \Had(SM, P)}, & F_\omega \in \F(SM, P)[[\hbar]]
\end{align*}
fulfilling
\begin{equation}
\label{eq:F_omega'}
 F_{\omega'} = \exp(\hbar \Gamma_{\omega_a'-\omega_a}) F_\omega.
\end{equation}
In particular, an element $F$ of $\A(SM, P)$ is entirely specified by $F_\omega$ for a single $\omega \in \Had(SM, P)$.
We can then equip $\A(SM, P)$ with the product
\[
 (F \star G)_\omega = F_\omega \star_\omega G_\omega.
\]
Note that the assignment $M \mapsto (\A(SM, P), \star)$ is a covariant functor from $\CatGSpMan$ to $\CatAlg$, which maps a morphism $\chi: (SM, P) \to (SM', P')$ to the morphism $\chi_*$ defined by
\begin{equation}
\label{eq:A_morphism}
 (\chi_* F)_{\omega'} = \chi_*(F_{\chi^* \omega'}),
\end{equation}
where on the \rhs $\chi_*$ is the morphism of $\F[[\hbar]]$. Furthermore, we define the algebra $\A_S(SM, P)$ of on-shell functionals analogously to $\F_S(SM, P)$. The local elements $\A_\loc(SM, P)$ of $\A(SM, P)$ are defined as those for which $F_\omega \in \F_\loc(SM, P)[[\hbar]]$ for one (and hence all) $\omega$. Again, $\A_\loc$ is a covariant functor from $\CatGSpMan$ to $\CatVec_i$.\footnote{In \cite{DHP09}, the Hadamard parametrix is used instead of Hadamard two-point functions. However, as discussed below, the parametrix is in general only defined in a neighborhood of the diagonal. Hence, the construction proposed in \cite{DHP09} does not work in general, i.e., one does not obtain a covariant functor to $\CatAlg$.}

It remains to show that Hadamard two-point functions exist. To this avail, we use the deformation argument of \cite{FNW81}. First of all, we have the following proposition, whose proof can be found in \ref{app:Deformation}.

\begin{proposition}
\label{prop:Deformation}
Let $M$ be globally hyperbolic, $(SM, \pi_S)$ a spin structure over $M$, $P$ a principal $G$ bundle over $M$ with connection $A$, $m \in C^\infty(M, \R)$, and $\Sigma$ a smooth Cauchy surface of $M$. There exist $M'$, $\tilde M$ globally hyperbolic and diffeomorphic to $M$ with spin structures $(SM', \pi'_S)$ and $(S\tilde M, \tilde \pi_S)$, $P'$, $\tilde P$ principal $G$ bundles over $M'$, $\tilde M$ with $P \simeq P' \simeq \tilde P$, connections $A'$, $\tilde A$, $m' \in C^\infty(M', \R)$, $\tilde m \in \R$, and smooth Cauchy surfaces $\Sigma', \tilde \Sigma' \subset M'$, $\tilde \Sigma \subset \tilde M$, $\Sigma' \cap \tilde \Sigma' = \emptyset$ such that

\begin{enumerate}
 \item $\Sigma$ and $\Sigma'$ are isometric and there are neighborhoods $U$, $U'$ of $\Sigma$, $\Sigma'$ such that $U$ and $U'$ are isometric and $m = i^* m'$ for this isometry. If $i_P$ is the bundle isomorphism $i_P: P \to P'$ we have $A|_{\pi_P^{-1}(U)} = i_P^* A'|_{\pi_{P'}^{-1}(U')}$. %Furthermore, $(SM|_U, \pi_S) \simeq (SM'|_{U'}, \pi'_S)$.
If $i_{SM}$ is the isomorphism $i_{SM}: SM \to SM'$, then $\pi_S' \circ i_{SM}|_{\pi_{SM}^{-1}(U)} = i_* \circ \pi_S|_{\pi_{SM}^{-1}(U)}$.

 \item $\tilde M$ is ultrastatic, i.e., $\tilde M = \R \times \tilde \Sigma$ with metric $\tilde g = - \ud t^2 \otimes \tilde h$, where $\tilde h$ is a Riemannian metric on $\tilde \Sigma$. The connection $\tilde A$ is time-invariant and has no time-component, i.e., $\Lie_{\del_t^*} \tilde A = 0$,
%\begin{align*}
% \Lie_{\del_t^*} \tilde A & = 0, &
% \tilde A(\del_t^*) & = 0,
%\end{align*}
where $\del_t^*$ is the horizontal lift \cite[Section~II.1]{KobayashiNomizu} of $\del_t$ \wrt $\tilde A$.

% \item $\tilde M$ is ultrastatic, i.e., $\tilde M = \R \times \tilde \Sigma$ with metric $g = - \ud t^2 \otimes h$, where $h$ is a Riemannian metric on $\tilde \Sigma$. Furthermore, $\tilde m(t,x) = m(x)$ in these coordinates. If $\{ \tilde U_i \}_{i \in I}$ is a covering of $\tilde \Sigma$ by charts, then on each $U_i = \R \times \tilde U_i$ there is a section $s_i$ of $P$ such that the pull-back $A_i$ of $A$ along $s_i$ is of the form
%\begin{align*}
% A_{i,0} & = 0, & A_{i,a}(t, x) & = A_{i,a}(x),
%\end{align*}
%in the natural coordinates on $U_i$. Analogously, on each $U_i$ there is a section $\tilde s_i$ of $S \tilde M$ such that the time component of the spin connection one-form pulled back by $\tilde s_i$ vanishes and the spatial ones are time-independent.

 \item $\tilde \Sigma$ and $\tilde \Sigma'$ are isometric and there are neighborhoods $\tilde U$, $\tilde U'$ of $\tilde \Sigma$, $\tilde \Sigma'$ such that $\tilde U$ and $\tilde U'$ are isometric and $\tilde m = \tilde \imath^* m'$ for this isometry. If $\tilde \imath_P$ is the bundle isomorphism $\tilde \imath_P: \tilde P \to P'$ we have $\tilde A|_{\pi_{\tilde P}^{-1} (\tilde U)} = \tilde \imath_P^* A'|_{\pi_{P'}^{-1} (\tilde U')}$.
%Furthermore, $(S \tilde M|_{\tilde U}, \tilde \pi) \simeq (S M'|_{\tilde U'},  \pi')$.
If $\tilde \imath_{SM}$ is the isomorphism $\tilde \imath_{SM}: S\tilde M \to SM'$, then $\pi_S' \circ \tilde \imath_{SM}|_{\pi_{S \tilde M}^{-1}(\tilde U)} = \tilde \imath_* \circ \tilde \pi_S|_{\pi_{S \tilde M}^{-1}(\tilde U)}$.

\end{enumerate}
\end{proposition}

On the ultrastatic spacetime $\tilde M$ and in the slicing $\tilde M \simeq \R \times \tilde \Sigma$, the Dirac equation may now be written as
\[
 - i \del_t \psi + K \psi = 0,
\]
where $K$ is, in a local trivialization, given by
\[
 K \psi \doteq i \gamma^0 \gamma^a (\del_a + \tilde \Omega_a - i \tilde A_a) \psi - i \gamma^0 m \psi,
\]
%where $\tilde \Omega$ is the spin connection and
with $1 \leq a \leq n-1$.
From \cite[Thm.~2.54]{BeemEhrlich} we conclude that $\tilde \Sigma$ is complete.
Then it follows from \cite[Thm.~2.2]{Chernoff}, that $K$ is an essentially self-adjoint operator on $L^2(\tilde \Sigma, D_\rho \tilde M|_{\tilde \Sigma})$ with domain $\Gamma^\infty_c(\tilde \Sigma, D_\rho \tilde M|_{\tilde \Sigma})$, where the scalar product is defined through the fiber-wise pairing
\[
 \skal{z_1 \otimes v_1}{z_2 \otimes v_2} = \bar{z_1} \cdot z_2 \skal{v_1}{v_2}_V.
\]
on $D_\rho \tilde M|_x \simeq \C^{2^{[n/2]}} \otimes V$.
Note that here we are not using spinor conjugation.\footnote{This scalar product stems from the standard inner product $(f, g) = - i \int_\Sigma \skal{f^+}{\gamma^\mu g} n_\mu$ for $f, g \in L^2(\Sigma, D_\rho M)$.} In the following we denote the self-adjoint extension of $K$ also by $K$. We can now proceed as in \cite{WrochnaDirac} to obtain distributional sections $\omega^\pm \in \Gamma^\infty_c(\tilde M^2, D_\rho^* \tilde M \boxtimes D_\rho \tilde M)'$. These are bisolutions fulfilling
\begin{align}
\label{eq:omegaWF}
 \WF \omega^\pm & \subset C_\pm, \\
\label{eq:omegaAnticommutator}
 \omega^+ + \omega^- & = i \tilde S,
\end{align}
where $\tilde S$ is the causal propagator on $(S \tilde M, \tilde P)$.
We can then define the distributional section $\omega \in \Gamma_c^\infty(\tilde M^2, D_\rho^\oplus \tilde M \boxtimes D_\rho^\oplus \tilde M)'$ by
\begin{align*}
 \omega(f', f) & \doteq \tfrac{1}{2} \left( \omega^+(f',f) + \overline{\omega^+(f^+, f'^+)} \right), \\
 \omega(f, f') & \doteq \tfrac{1}{2} \left( \omega^-(f',f) + \overline{\omega^-(f^+, f'^+)} \right), \\
 \omega(f, g) & \doteq 0, \\
 \omega(f', g') & \doteq 0,
\end{align*}
where $f, g \in \D(S \tilde M, \tilde P)$, $f', g' \in \D^*(S \tilde M, \tilde P)$. Then \eqref{eq:HadamardWaveEq} follows from $\omega^\pm$ being bi-solutions, \eqref{eq:HadamardConjugation} follows by definition, and \eqref{eq:HadamardWF} follows from \eqref{eq:omegaWF}. Condition \eqref{eq:HadamardAnticommutator} is a consequence of \eqref{eq:omegaAnticommutator}. Hence, $\omega$ is a Hadamard two-point function on $\tilde M$.\footnote{An equivalent approach for the construction of a Hadamard two-point function on $\tilde M$ would be to consider the CAR-algebra corresponding to the above Hilbert space (supplemented by co-spinorial sections) and using the projection on the positive spectrum of $K$ to define a state \cite{ArakiCAR}. The corresponding two-point function fulfills the wave front condition, by \cite{SahlmannVerchPassivity}.}

It remains to transport $\omega$ to $M$. By the isometry of a neighborhood $\tilde U$ of $\tilde \Sigma$ and a neighborhood $\tilde U'$ of $\tilde \Sigma'$, we can push-forward $\omega|_{\tilde U \times \tilde U}$ to a distribution on $\tilde U' \times \tilde U'$. Using the equation of motion, we extend it to the entire $M' \times M'$. By the isometry of neighborhoods $U'$, $U$ of $\Sigma'$ and $\Sigma$, we may transfer it to $M$ and again use the equation of motion there to extend it to $M \times M$. Due to the coincidence of Cauchy data, it is clear that the symmetric part still coincides with the fundamental solution. It remains to show that the Hadamard property is conserved under the extension procedure. By the propagation of singularity theorem, one only has to show that no elements $(x, \xi; y, 0)$ or $(x, 0; y, \eta)$ may appear in the wave front set. As the two-point function $\omega$ gives rise to a quasi-free state on the Cauchy data on $\tilde \Sigma$, one may use the calculus of Hilbert space valued distributions and argue as in \cite[Sec.~4.2]{SandersDirac} to show that the wave front set may not contain such elements. We have thus proven:
\begin{theorem}
There exist Hadamard two-point functions on each $\F(SM, P)$.
\end{theorem}

\subsection{Fields}
\label{sec:Fields}

In the setting of local covariant field theories,
fields are objects defined on all backgrounds simultaneously, in a coherent way \cite{BrunettiFredenhagenVerch}. In the categorical language, this is encoded in requiring that they are natural transformations $\Phi: \T_c \to \A_\loc$, where $\T_c$ was defined in \eqref{eq:Tens}. An example are the linear fields
\begin{equation}
\label{eq:LinearField}
 \psi_{(SM, P)}(u)_\omega(B) \doteq \int \skal{u}{B_1}(x) \ud_g x, \quad u \in \T_c^{1 0}(SM,P) =\D^\oplus(SM, P),
\end{equation}
which are natural transformations $\T_c^{10} \to \A_\loc$. We note that there is no dependence on $\omega$ on the r.h.s., as all the operators $\Gamma_{\omega_a' - \omega_a}$, \cf \eqref{eq:F_omega'}, vanish on this functional, since it is linear in the configuration. We also note that it fulfills
\begin{align*}
 \psi_{(SM, P)}(u)^* & = \psi_{(SM, P)}(u^+) \\
 \psi_{(SM, P)}(u) \star \psi_{(SM, P)}(v) + \psi_{(SM, P)}(v) \star \psi_{(SM, P)}(u) & = i \hbar S^\oplus_{(SM, P)}(u,v).
\end{align*}
By choosing $u$ to be a pure cospinor (spinor), one obtains the usual spinor (cospinor) fields, which, in an abuse of notation, will be denoted by $\psi$ and $\psi^+$ in Section~\ref{sec:Currents}.

\begin{remark}
\label{rem:GaugeCovariance}
The fields we consider are in general not gauge invariant, but gauge covariant, in the sense that we may integrate configurations with test sections (elements of $\T_c$) that transform nontrivially under the gauge group action, i.e., morphisms of the form \eqref{eq:GaugeTrafo}.\footnote{A morphism of the form \eqref{eq:GaugeTrafo} induces an isomorphism $\T_c(\chi_P): \T_c(SM,P) \to \T_c(SM',P')$. But we had $P'=P$, $SM' = SM$, so this yields an automorphism of $\T_c(SM, P)$. By a nontrivial transformation under the gauge group, we mean that this automorphism does not act as the identity.} In this respect we differ from the setting of \cite{HollandsYM}, where the ``local and covariant functionals'' are required to be gauge invariant, \cf Section~2.1 there.
\end{remark}

In contrast, the definition of nonlinear fields (Wick powers), i.e., natural transformations $\T_c^{j A} \to \A_\loc$ for $j>1$, is not straightforward. The problem is to define them on all backgrounds and at the same time fulfill the relations \eqref{eq:F_omega'} and \eqref{eq:A_morphism}. The crucial point is to find a trivializing distribution $H$ which is covariantly assigned to each background and is such that $\omega - H$ is smooth for all Hadamard two-point functions. These are the parametrices, which we define as follows:

\begin{definition}
\label{def:Parametrix}
A \emph{parametrix} $H$ is a quasi-covariant assignment $(SM, P, A, m) \to H \in \Gamma_c^\infty(U, D^\oplus_\rho M \boxtimes D^\oplus_\rho M)'$, where $U$ is a neighborhood of the diagonal of $M \times M$, such that \eqref{eq:HadamardAnticommutator}, \eqref{eq:HadamardConjugation}, \eqref{eq:HadamardWF} hold.
\emph{Quasi-covariance} here means that for $\chi: D_\rho^\oplus M \to D_\rho^\oplus M'$ the bundle morphism corresponding to a morphism $(SM, P, A, m) \to (SM', P', A', m')$ we have that $H - \chi^* H'$ is smooth on the common domain and vanishing at the diagonal, together with all the derivatives.
\end{definition}

We note that the choice of the domain $U$ is irrelevant, as for our purposes $H$ only needs to be known in an arbitrarily small neighborhood of the diagonal. The requirement of quasi-covariance is crucial for the constructions presented below to be covariant. To our opinion, this aspect is not properly emphasized in \cite{DHP09}, at least not explicitly. 
A consequence of the definition is the following:
\begin{proposition}
The difference $H - \omega$ is smooth on the domain $U$ for any Hadamard two-point function $\omega$ and any parametrix $H$.
\end{proposition}
This is basically Lemma 2.9 of \cite{SandersHadamard}. For convenience, we include a proof. 
\begin{proof}
The distributional sections $\omega$ and $H$ share the same symmetric part, i.e., $\omega_s - H_s = 0$, where $\omega_s(u,u') \doteq \frac{1}{2}(\omega(u,u') + \omega(u',u))$. We also know that $\WF(\omega-H) \subset C_+$. Assume that $p \in C_+$ is contained in $\WF(\omega-H)$. As the distribution $(u, u') \mapsto \omega(u', u)$ has wave front set contained in $C_-$, and analogously for $H$, it follows that $p$ is also contained in $\WF(\omega_s - H_s)$, as it can not be cancelled by symmetrization of the distribution. But $\WF(\omega_s - H_s)$ is empty, so $\omega - H$ is smooth.
\end{proof}

\begin{remark}
Since Hadamard two-point functions exist, as proven above, it follows that a parametrix is a bi-solution up to smooth terms. Alternatively, one may argue as in the Note Added in Proof in \cite{Radzikowski}.
\end{remark}

With a parametrix $H$, we may associate to a local functional $F \in \F_\loc$ an element of $\A_\loc$ by
\begin{equation}
\label{eq:F_to_A}
 (F_H)_\omega \doteq \exp(\hbar \Gamma_{\omega - H}) F.
\end{equation}
This is well-defined as $H - \omega$ is smooth and the values of all its derivatives on the diagonal are unambiguous. As we only act on local functionals, the expression is well-defined even though $H$ is only defined in a neighborhood of the diagonal. There is a canonical natural transformation $\Psi: \T_c \to \F_\loc$, defined by
\begin{equation}
\label{eq:Psi}
 \Psi_{(SM, P)}(t)(B) \doteq \sum_{k=0}^\infty \int \skal{t^{\underline{\mu_1} \dots \underline{\mu_k}}}{\nabla^{\oplus 1}_{(\underline{\mu_1})} \dots \nabla^{\oplus k}_{(\underline{\mu_k})} B_k}(x) \ud_g x,
\end{equation}
where $\underline{\mu_i}$ are multiindices and $\nabla_{(\underline{\mu_i})}^i$ denotes the symmetrized covariant derivative on the $i$th coordinate, with $\nabla^\oplus \doteq \nabla \oplus \nabla^*$. Composing $\Psi|_{\T_c^{j A}}$ with the map \eqref{eq:F_to_A}, we obtain fields, called the \emph{Wick powers}.
Hence, given a parametrix, a plethora of fields is available.

In order to show that parametrices exist, let us first review the construction of the causal propagator. In order to get rid of the first order term in $P$, \cf \eqref{eq:P}, we introduce a new covariant derivative $\tilde \nabla_\mu \doteq  \nabla_\mu -  m \gamma_\mu$. Then we have
\[
 P = g^{\mu \nu} \tilde \nabla_\mu \tilde \nabla_\nu + \tfrac{1}{4} [\gamma^\mu, \gamma^\nu] ({\mathfrak{R}}_{\mu \nu} - i F_{\mu \nu}) - (n -1) m^2.
\]
Analogously, we proceed with $P^* = D^* D^*$, by using $\tilde \nabla^*_\mu \doteq \nabla^*_\mu + m \gamma_\mu$.

On each causal domain $\Omega$, i.e., a geodesically convex domain which is globally hyperbolic, the \emph{Hadamard coefficients} $V_k \in \Gamma^\infty(\Omega \times \Omega, D_\rho M \times D^*_\rho M)$ are recursively defined by the transport equation 
\begin{equation*}
%\label{eq:TransportEquation}
 \tilde \nabla_{\del \Gamma} V_k - \left( - \tfrac{1}{2} \nabla^\mu \del_\mu \Gamma - n + 2k \right) V_k = 2 k P V_{k-1}, 
\end{equation*}
with the initial condition $V_0(x, x) = \id_{D_\rho M_x}$. Here all derivatives act on the first coordinate and $\Gamma(x, x')$ is the negative of the squared geodesic distance along the unique geodesic connecting $x$ and $x'$. The transport equation defines the Hadamard coefficients locally and covariantly. Analogously, one defines the Hadamard coefficients for $P^*$. The retarded/advanced propagator for $P$ can now be approximated on $\Omega \times \Omega$ up to a smooth section $r_{\ret / \adv}$ \cite[Sec.~2.4]{BGP07},
\begin{equation}
\label{eq:DeltaApprox}
 \Delta_{\ret/\adv}(x,x') - \sum_{j=0}^{\infty} \chi(\Gamma(x,x') / \eps_j) V_j(x,x') R_\pm(2+2j)(x,x') = r_{\ret/\adv}(x,x'),
\end{equation}
where $r_{\ret / \adv}(x,x')$ vanishes unless $x$ is in the causal future/past of $x'$.
Here $\chi : \R \to \R$ is a smooth compactly function identical to $1$ on $[-1,1]$, and the sequence $\{ \eps_j \}$ of positive reals is chosen to ensure convergence. The distributions $R_\pm(j)$ are so-called Riesz distributions, whose singular support is the light cone.

Likewise, there are parametrices $h_\pm$ for $P$ on $\Omega \times \Omega$, given by 
\begin{equation*}
%\label{eq:hFormal}
 h_\pm(x,x') \doteq \frac{1}{2 \pi} \sum_{j=0}^\infty \chi(\Gamma(x,x') / \eps_j) V_j(x,x') T_\pm(2+2j)(x,x'),
\end{equation*}
where $T_\pm(j)$ are certain distributions, which for $j \in \{ 0, 2, 4, \dots \}$ fulfill
\begin{align}
\label{eq:TR}
 T_+(j) - T_-(j) & = 2 \pi i (R_+(j) - R_-(j)), \\
\label{eq:WF_T}
 \WF(T_\pm(j)) & \subset C_\pm.
\end{align}
Furthermore, for $2j \geq n$, we have
\begin{equation*}
 T_\pm(2j) = c_{2j} \Gamma^{[j-n/2]} \begin{cases} \log \Gamma_{\pm \eps}/ \Lambda & n \text{ even,} \\ \Gamma_{\pm \eps}^{1/2} & n \text{ odd,} \end{cases}
\end{equation*}
where $\Lambda$ is a fixed length scale and $\Gamma_\eps$ is $\Gamma$ equipped with a suitable $i \eps$ description at $x = x'$. The singular behavior stems entirely from $\log \Gamma_{\pm \eps}$ or $\Gamma_{\pm \eps}^{1/2}$. Now the $\eps_j$ may be chosen such that, for $N \geq n/2$,
\[
 \sum_{j = N}^\infty c_{2j+2} \chi(\Gamma/\eps_j) \Gamma^{[j+1-n/2]} V_j 
\]
converges in $C^k$ for all $k$ \cite[Lemma~2.4.2]{BGP07}, i.e., it is smooth. Note that when evaluating derivatives of this expression at coinciding points $x = x'$, only a finite number of terms are nonzero, and these are independent of the $\{ \eps_j \}$. It follows that by changing the $\{ \eps_j \}$, one does not change the coinciding point limit of derivatives of this expression. This ensures the quasi-covariance of the construction. 
Also note that we may choose the same sequence $\{ \eps_j \}$ as in \eqref{eq:DeltaApprox}.
Hence, by \eqref{eq:WF_T}, we have $\WF (h_\pm) \subset C_\pm$.
%By construction and \eqref{eq:WF_T}, $h_\pm$ fulfill \eqref{eq:HadamardWF}.
Furthermore, by \eqref{eq:TR} and \eqref{eq:DeltaApprox}, 
\[
 r \doteq h_+ - h_- - i \Delta
\]
is smooth, where $\Delta \doteq \Delta_\ret - \Delta_\adv$. Define $\tilde h_\pm \doteq h_\pm \mp \frac{1}{2} r$ so that $\tilde h_+ - \tilde h_- = i \Delta$. Note that, due to the support properties of $r_{\ret / \adv}$,  $r$  vanishes, together with all derivatives, at the diagonal. By covering $M$ with causal domains $\Omega_i$ define the neighborhood $U \doteq \cup_i (\Omega_i \times \Omega_i)$ of the diagonal and choose a corresponding partition of unity $\chi_i$ of $U$. Then define $\tilde h_\pm$ on $U$ by $\tilde h_\pm \doteq \sum_i \chi_i \tilde h_{i \pm}$, where the $\tilde h_{i\pm}$ are constructed as described above.

We recall that the retarded/advanced propagators $\Delta^{(*)}_{\ret/\adv}$ for $P^{(*)}$ are related by \cite[Lemma~3.4.4]{BGP07}
\[
 \int \skal{f'}{\Delta_{\ret/\adv} f}(x) \ud_{ g} x = \int \skal{\Delta^*_{\adv/\ret} f'}{f}(x) \ud_{ g} x.
\]
It follows that the corresponding causal propagators are related as
\[
 \int \skal{f'}{\Delta f}(x) \ud_{ g} x = - \int \skal{\Delta^* f'}{f}(x) \ud_{ g} x.
\]
Hence, the distributions $\tilde h^*_\pm \in \Gamma^\infty_c(U, D_\rho M \times D_\rho M^*)'$ defined by
\begin{equation}
\label{eq:h*}
 \tilde h^*_\pm(f,f') \doteq \tilde h_\mp(f',f)
\end{equation}
fulfill $\WF (\tilde h^*_\pm ) \subset C_\pm$ and $\tilde h^*_+ - \tilde h^*_- = i \Delta^*$.

As discussed in Section~\ref{sec:DiracOperator}, the retarded/advanced propagators defined by \eqref{eq:DefS_ra} fulfill $S_{\ret / \adv}^{(*)} \circ D^{(*)} = \id$. 
In particular,
\[
 D^{(*)} \circ \Delta^{(*)}_{\ret/\adv} \circ D^{(*)} = \id.
\]
Hence, we could also define the retarded/advanced propagator as
\begin{equation*}
 S^{(*)}_{\ret/\adv} = \Delta^{(*)}_{\ret/\adv} \circ D^{(*)},
\end{equation*}
but as it is unique, the two definitions coincide. A parametrix $H \in \Gamma^\infty_c(U, D_\rho^\oplus M^2)'$ for $D^\oplus$ can now be defined as
\begin{align*}
 H(f',f) & \doteq \tfrac{1}{4} \left( \tilde h_+(D^* f', f) + \tilde h_+(f', D f) + \overline{ \tilde h_-(f^+, D f'^+)} + \overline{ \tilde h_-(D^* f^+, f'^+)} \right), \\
 H(f,f') & \doteq - \tfrac{1}{4} \left( \tilde h_-(D^* f', f) + \tilde h_-(f', D f) + \overline{ \tilde h_+(f^+, D f'^+)} + \overline{ \tilde h_+(D^* f^+, f'^+)} \right), \\
 H(f,g) & \doteq 0, \\
 H(f',g') & \doteq 0,
\end{align*}
where $f,g \in \Gamma_c^\infty(M, D_\rho M)$, $f',g' \in \Gamma_c^\infty(M, D_\rho^* M)$. Note that, by the above discussion on the retarded/advanced propagator,
\begin{align*}
 \tilde h_+(D^* f', f) - \tilde h_-(D^* f', f) & = i \Delta(D^* f', f) = i S(f', f) \\
 &= i \Delta(f', D f) = \tilde h_+(f', D f) - \tilde h_-(f', D f)
\end{align*}
so $H$ has the anticommutator property. Hence, we have shown:
\begin{proposition}
Parametrices exist.
\end{proposition}

\begin{remark}
The construction differs from constructions in the literature, \cf \cite{SahlmannVerchHadamard, DHP09}, by the fact that we do not use an auxiliary operator $\tilde D = \gamma^\mu \nabla_\mu + m$ to define $P$, so that our $P$ has first order terms that we have to deal with by a change of the connection. The advantage of our construction is that we may use an average $\tilde h_+ (D^* f', f) + \tilde h_+(f', D f)$ in the definition of $H$, which facilitates the proof of current conservation, \cf Section~\ref{sec:CurrentConservation}. A similar construction with auxiliary operators would require $D$ and $\tilde D$ to commute, which is only the case if $m$ is constant.
\end{remark}

\begin{remark}
\label{rem:Scaling}
%For even $n$, the distributions $T_\pm$ involve a term of the form $\log \Gamma_\eps / \Lambda^2$, where $\Lambda$ is a length scale that has to be introduced to make the logarithm well-defined.\footnote{In \cite{DHP09} it is proposed to choose $\Lambda$ proportional to the inverse mass $m^{-1}$. This only works if the mass is constant and non-zero. In particular, this prescription violates the smoothness condition introduced below (adapted such that $m$ is required to be constant).}
The length scale $\Lambda$ that has to be introduced in $T_\pm$ for even $n$ is arbitrary, but has to be fixed to the same value on all backgrounds.\footnote{In \cite{DHP09} it is proposed to choose $\Lambda$ proportional to the inverse mass $m^{-1}$. This only works if the mass is constant and non-zero. In particular, this prescription violates the smoothness condition introduced below (adapted such that $m$ is required to be constant).} The need for such a scale plays an important role in the discussion of the axioms for time-ordered products in the following subsection and of the scaling behavior in Section~\ref{sec:Scaling}.
\end{remark}

\begin{remark}
\label{rem:WickUniqueness}
The parametrix, and hence the Wick powers, is not unique. One may always modify the parametrix by a smooth, locally and covariantly constructed function. In Section~\ref{sec:Currents}, we elaborate on this, and show that this freedom may be used to achieve a conserved stress-energy tensor. In the present setting, by modifying the parametrix, one modifies Wick squares and all higher order powers. For the scalar field, Hollands and Wald also allowed for redefinitions of the Wick powers that only affect the $k$th and higher order powers, for an arbitrary $k$ \cite{HollandsWaldWick}. To achieve this in the present setting, one would have to add to $\hbar \Gamma_{\omega - H}$ in \eqref{eq:F_to_A} operators of the form
\[
 \hbar^{k/2} \Gamma^k_{H_k} F \doteq \hbar^{k/2} \int H_k(x_1, \dots, x_k)  F^{(k)}(x_1, \dots, x_k) \ud_gx_1 \dots \ud_gx_k,
\]
where $H_k$ is smooth, locally and covariantly constructed, and defined in a neighborhood of $\diag^k$. But as such redefinitions are not necessary for the fulfillment of current and stress-energy conservation, we do not pursue this issue further.
\end{remark}

\subsection{Time-ordered products}

We now discuss the construction of renormalized time-ordered products in our setting. Let us start with the following definition:
\begin{definition}
The vector space $\MT_c^k$ of $k$-\emph{local test tensors} is defined as the $\Z/2$-graded $k$-fold tensor product of $\T_c$, where the grade of $B \in \T_c^{j A}$ is $\betrag{B} = j \mod 2$. A typical element is denoted by $B_1 \gt \dots \gt B_k$, where the hat indicates the graded tensor product. The vector space $\MT_c$ of \emph{multilocal test tensors} is defined as the direct sum of the $\MT_c^k$.
\end{definition}

We may now introduce the notion of multilocal fields. 
\begin{definition}
A \emph{multilocal field} is a natural transformation
\[
 \Phi: \MT_c \to \A,
\]
where, by composition with the forgetful functor, we interpret $\A$ as a functor between $\CatGSpMan$ and $\CatVec_i$.
\end{definition}
Obviously, this is a generalization of the notion of fields as introduced in Section~\ref{sec:Fields}. Often we will want to be more specific, and denote $\Phi_k^{\underline{j} \underline{A}}$ the induced natural transformation
\[
 \Phi_k^{\underline{j} \underline{A}} : \T_c^{j_1 A_1} \gt \dots \gt \T_c^{j_k A_k} \to \A.
\]
Here $\underline{j}$ and $\underline{A}$ are the multiindices containing the $j_i$, $A_i$. We recall that $j$ stands for the number of fields, and $A \in \N_0^j$ for the number of derivatives on the separate fields.  

There are several further conditions on time-ordered products. In order to formulate these, we introduce the concept of scaling. The idea is to relate the theory on the background $(SM, P, A, m)$ with the theory on another background $(SM', P', A', m')$, where $SM$ ($P$) and $SM'$ ($P'$) are isomorphic as principal $\Spin_0$ ($G$) bundles, and only the geometric data changes. Using these isomorphisms, we can identify configurations and test tensors on the two backgrounds. The nontrivial step is the setup of the isomorphism of $SM$ and $SM'$. For this, we proceed as follows: In deforming $SM$ to $SM'$ we keep the $\Spin_0$ bundle and only change the spin projection $\pi_S$. For that, we identify $FM$ with a principal $\Lor_0$ bundle $LM$. To construct $SM'$, we keep $LM$ and the projection from $SM$ to $LM$, but change the identification of $FM$ and $LM$. It is given by a vielbein, which we denote in local coordinates and some trivialization of $LM$ by $e^\mu_a$. Infinitesimally, we now translate a change of $g^{\mu \nu}$ into a change of $e^\mu_a$ by $\delta e^\mu_a = - \frac{1}{2} e^\nu_a g^{\mu \lambda} \delta g_{\nu \lambda}$. This corresponds to the method used in \cite{ForgerRomer} to compute the stress-energy tensor of Dirac fields.

Let us now explicitly construct a scaled background $(SM', P', A', m')$.
%For that, we restrict to $M$ with trivial topology, i.e., we consider the local theory.
We set $P' = P$, $A' = A$, and $SM' = SM$ (as a $\Spin_0$ bundle). In local coordinates, define\footnote{In the language of \cite{ForgerRomer}, this means that the fields transform according to their Weyl dimension.} $g'_{\mu \nu} = \lambda^{-2} g_{\mu \nu}$, $m' = \lambda m$. According to the above, this means ${e'}^\mu_a = \lambda^{-1} e^\mu_a$ for the vielbein.
%Also note that due to the triviality of $M$, this defines a connection $A'$ on $P'$, \cf \cite[Prop.~II.1.4]{KobayashiNomizu}.
Clearly, this transformation simply scales the Dirac operator. Analogously, the fundamental solutions, and hence also the Hadamard two-point functions scale. There is thus a $*$-isomorphism $\sigma_\lambda: \A(SM', P') \to \A(SM, P)$, acting on linear fields as
\begin{equation}
\label{eq:Def_sigma}
 \sigma_\lambda (\psi_{(SM',P')}(u))_\omega = \lambda^{-\frac{n+1}{2}} \psi_{(SM,P)}(u)_{\omega_\lambda},
\end{equation}
where $\omega_\lambda(u,v) = \lambda^{-n-1} \omega(u,v)$, \cf \cite[Lemma~4.2]{HollandsWaldWick} for a proof in the scalar case.
Note that here we used the identification of sections of $D^\oplus_\rho M$ and $D^\oplus_\rho M'$ induced by the bundle isomorphisms constructed above.
For a multilocal field $\Phi_k$, one may define another multilocal field $S_\lambda \Phi_k$ by
\[
 (S_\lambda \Phi_k)_{(SM, P)}(t) \doteq \lambda^{nk} \sigma_\lambda ({\Phi_k}_{(SM', P')}(\chi^* t)),
\]
where $t \in \MT_c^k$ and $\chi^*$ is the pullback to the scaled background.
The \emph{scaling dimension} of a field $\Phi_k^{\underline{j} \underline{A}}$ is defined as
\[
 d_{\Phi_k^{\underline{j} \underline{A}}} = \sum_{i = 1}^k \left( \tfrac{n-1}{2} j_i + \betrag{A_i} \right).
\]

The time-ordered products are now multilocal fields that fulfill further axioms. First of all, we require them to be well-defined as natural transformations
 \begin{equation}
\label{eq:TO_F_loc}
  T_k: \underbrace{\F_\loc \gt \dots \gt \F_\loc}_{k \text{ times}}  \to \A,
 \end{equation}
obtained by using $\Psi$, \cf \eqref{eq:Psi}, to map the elements of $\T_c$ to $\F_\loc$. Again, $\hat \otimes$ denotes the $\Z / 2$-graded tensor product where the grading refers to the grade of $F \in \F_\loc$ modulo $2$. Due to the integration, this induces relations between time-ordered products with different numbers of derivatives, called the Leibniz rule in \cite{HollandsWaldStress} and the Action Ward Identity in \cite{DuetschFredenhagenAWI}. 
In order to formulate it, we introduce a notation that will also be useful later on. The time-ordered product $T_k^{\underline{j} \underline{A}}$ may be seen as an $\A$-valued distributional section. Given a local trivialization, we write its integral kernel as 
\[
 T^{\underline{\alpha_1} \dots \underline{\alpha_k}}(x_1, \dots, x_k),
\]
where the $\underline{\alpha_i}$ are multiindices consisting of tuples $(a_l, \underline{\mu}_l)_{l \in \{ 1, \dots, j_i\}}$, where the $a_l$ are spinorial and gauge indices and the $\underline{\mu_l}$ spacetime multiindices with $\betrag{\underline{\mu_l}} = A_i(l)$. The Leibniz rule can then be formulated as
\[
 \nabla_i^\mu T^{\underline{\alpha_1} \dots \underline{\alpha_k}}(x_1, \dots, x_k) = \sum_{l=1}^{j_i} T^{\underline{\alpha_1} \dots (\underline{\alpha_i} +_l \mu) \dots \underline{\alpha_k}}(x_1, \dots, x_k) + \dots,
\]
where $\underline{\alpha} +_l \mu$ means adding $\mu$ to the multiindex $\underline{\mu_l}$ inside $\underline{\alpha}$, and the dots stand for lower order terms obtained by symmetrizing the derivatives.

There are a couple of further conditions:

\begin{description}
 \item[Support:]
 The support of $T_k(t)$, $t \in \MT_c^k$, \cf \eqref{eq:suppF}, is contained in $\supp_M t$, defined as
\[
  \supp_M t  = \{ x | (x, x_2, \dots x_k) \in \supp t \}.
\]
 \item[Causal factorization:] Let $t \in \MT_c^k$, $t' \in \MT_c^l$ be multilocal test sections such that $\supp_M t$ has no intersection with the past of $\supp_M t'$. Then
 \[
  T_{k+l}(t \gt t') = T_k(t) \star T_{l}(t').
 \]
 \item[Scaling:] The time-ordered products $T_k^{\underline{j} \underline{A}}$ scale \emph{almost homogeneously}, i.e., there are natural numbers $c_k^{\underline{j} \underline{A}}$ such that
\begin{equation}
\label{eq:Scaling}
 \left( \lambda \frac{\del}{\del \lambda} - d_{T_k^{\underline{j} \underline{A}}} \right)^{c_k^{\underline{j} \underline{A}}} S_\lambda T_k^{\underline{j} \underline{A}} = 0.
\end{equation}
 \item[Microlocal spectrum condition:]  Let $\omega$ be a quasi-free Hadamard state on $\A(SM, P)$. Then the wave front set of the distributional section $\omega(T^{\underline{\alpha_1} \dots \underline{\alpha_k}}(x_1, \dots, x_k))$ is contained in $C_T^k \subset T^* M^k$, defined through decorated graphs, \cf \cite{BrunettiFredenhagenScalingDegree, HollandsWaldTO}.
 \item[Smoothness:] The time-ordered products depend smoothly on the background fields. Thus, let $g_s, A_s, m_s$ depend smoothly on a parameter $s \in \R$. Let $\omega^s$ be a family of Hadamard states on $\A(SM^{(s)}, P^{(s)})$, with smooth truncated $n$-point functions that depend smoothly on $s$.
One then requires that 
\begin{multline*}
 \WF \left( \omega^{(s)} \left( T^{(s)}_k(x_1, \dots, x_k) \right) \right) \\ \subset \left\{ (s, \sigma; \{ x_i, \xi_i \}) \in \dot T^*(\R \times M^k) | (\{x_i, \xi_i \}) \in C_T^{k, (s)} \right\}.
\end{multline*}

 \item[Analyticity:] In the case of an analytic spacetime, the Wick products depend analytically on the background fields. This is made precise by a condition analogous to the one for smoothness.
\end{description}
There are further conditions which are most easily stated for time-ordered products interpreted as maps \eqref{eq:TO_F_loc}. However, it is clear that these can be reformulated for time-ordered products interpreted as multilocal fields.
\begin{description}
\item[Expansion:] The time ordered product commutes with functional differentiation, i.e.
\begin{equation}
\label{eq:TO_Expansion}
 T(F_1 \gt \dots \gt F_k)^{(1)}(x) = \sum_{i = 1}^k (-1)^{\sum_{l=1}^{i-1} \betrag{F_l}} T(F_1 \gt \dots \gt F_i^{(1)}(x) \gt \dots \gt F_k).
\end{equation}
\item[Unitarity:] We have
\[
T(F_1 \gt \dots \gt F_k)^* = \sum_{I_1 \sqcup \dots \sqcup I_j} (-1)^{n+j+\Pi} T( \hat \bigotimes_{i \in I_1} F_i^*) \star \dots \star T( \hat \bigotimes_{i \in I_j} F_i^*),
\]
where $I_1 \sqcup \dots \sqcup I_j$ denotes all partitions of $\{ 1, \dots, k \}$ into nonempty, pairwise disjoint subsets. $\Pi$ denotes a combinatorial factor, depending on the grades of the $F_i$ and the partition, which accounts for the reordering of the $F_i$ on the right hand side.
\item[Equation of motion:] If $\psi$ denotes the linear field \eqref{eq:LinearField}, then
\begin{equation}
\label{eq:TO_eom}
 T( \psi(D^\oplus u) \gt F_1 \gt \dots \gt F_k) = i \skal{T(\hat \bigotimes_i F_i)^{(1)}}{u} + \psi(D^\oplus u) \star T(\hat \bigotimes_i F_i).
\end{equation}
\end{description}

The time-ordered products of order 1 are simply the Wick powers, as defined by \eqref{eq:F_to_A}. As noted above, \cf Remark~\ref{rem:WickUniqueness}, these are not unique.

The rationale behind the axiom of almost homogeneous scaling is the following: Because the classical theory has homogeneous scaling, one would like to impose this condition also for the quantum theory. However, as discussed in Remark~\ref{rem:Scaling}, the parametrix contains a logarithmic term for $n$ even, which necessitates the choice of a scale. This breaks homogeneous scaling, and almost homogeneous scaling is the minimal generalization of homogeneous scaling such that Wick products exist. Also the extension of distributions necessary to define time-ordered products typically breaks scale invariance.

Due to the axiom of causal factorization, time-ordered products can be defined recursively, by extension of distributional sections defined on $M^k \setminus \diag^k$ to $M^k$ \cite{BrunettiFredenhagenScalingDegree}. The important point is to ensure locality and local Lorentz and gauge covariance in this extension, to preserve the functoriality. For the scalar field, this was performed in \cite{HollandsWaldTO}, see also \cite{HollandsYM}. In the following, we only describe the changes to the argument that are necessary to accomodate charged spinors.

Due to the Leibniz rule, the distributional sections $T^{\underline{j} \underline{A}}_k$ are not independent. The action of the derivation defines the subspace of the Leibniz dependent ones. As in \cite{HollandsWaldTO}, we may choose a complement of this subspace and only have to define the time-ordered products on a basis of this complement.

One considers a small enough neighborhood $U$ of a point $(x, \dots, x)$ on the diagonal $\diag^k$,
and expands a time-ordered product $T_0$ defined up to $\diag^k$ into Hadamard-ordered ones, i.e.,
\begin{multline*}
 T_0^{\underline{\alpha_1} \dots \underline{\alpha_k}}(x_1, \dots, x_k)_\omega  \\
 = \sum_{\underline{\beta}_i \subset \underline{\alpha}_i} c_{\underline{\alpha} \underline{\beta}} t_0^{\underline{\beta}_1 \dots \underline{\beta}_k}(x_1, \dots, x_k) \WDp{\Psi^{\underline{\alpha}_1 \setminus \underline{\beta}_1}(x_1) \dots \Psi^{\underline{\alpha}_k \setminus \underline{\beta}_k}(x_k)}_{\omega, H},
\end{multline*}
where the $c$'s are combinatorical constants, $t_0$ a distributional section, and
\[
 \WDp{\Psi^{\underline{\alpha}_1}(x_1) \dots \Psi^{\underline{\alpha}_k}(x_k)}_{\omega, H} \doteq \exp(\hbar \Gamma_{\omega-H}) \Psi^{\underline{\alpha}_1}(x_1) \dots \Psi^{\underline{\alpha}_k}(x_k).
\]
Here $\Psi^{\underline{\alpha}}(x)$ denotes the integral kernel of the map $\Psi$, \cf \eqref{eq:Psi}, interpreted as an $\F_\loc$-valued distributional section. 
The form of the above expansion follows from \eqref{eq:TO_Expansion}, \cf the discussion in \cite{HollandsWaldTO} for the scalar case.
Because of \eqref{eq:A_morphism} and the requirement on the parametrix, the distributions $t_0$ are gauge invariant in the following sense: The difference
\begin{multline}
\label{eq:t_0_gaugeInv}
 \rho(g(x_1))^{\underline{\alpha}_1}_{\underline{\beta}_1} \dots \rho(g(x_k))^{\underline{\alpha}_k}_{\underline{\beta}_k} t_0^{\underline{\beta}_1 \dots \underline{\beta}_k}[g A_\mu g^{-1} + g \del_\mu g^{-1}](x_1, \dots, x_k) \\ - t_0^{\underline{\alpha}_1 \dots \underline{\alpha}_k}[A_\mu](x_1, \dots, x_k)
\end{multline}
is smooth and vanishes, with all its derivatives, at the diagonal.

To extend $t_0$ to all of $U$, one proceeds as follows: Fix the last coordinate to $x$ and describe the coordinates $x_1, \dots, x_{k-1}$ by Riemannian normal coordinates $\xi_1, \dots, \xi_{k-1}$ \wrt $x$. It then suffices to extend the resulting distribution on $\R^{n(k-1)} \setminus \{ 0 \}$ to the origin. To do this in a local way, one performs a scaling expansion of $t_0$. In Riemannian normal coordinates and in a given trivialization, one defines the following family of metrics, masses, and gauge potentials:\footnote{For a description how to identify the sections of bundles for different geometric data, we refer to the discussion preceding the introduction of the scaling transformation in Section~\ref{sec:Fields}.}
\begin{align*}
 g^{(s)}_{\mu \nu}(\xi) & \doteq g_{\mu \nu}(s \xi), & m^{(s)}(\xi) & \doteq s m(s \xi), & A_\mu^{(s)}(\xi) & \doteq s A_\mu(s \xi). 
\end{align*}
Now one Taylor expands $t_0$ around $s = 0$, i.e.,
\[
 t_0 = \sum_{l=0}^p \frac{1}{l!} \tau_{0, l} + r_{0, p}
\]
with
\begin{align*}
 \tau_{0, l}(\cdot, x) & \doteq \frac{\ud^l}{\ud s^l} \left. t_0[g^{(s)}, m^{(s)}, A^{(s)}](\cdot, x) \right|_{s = 0}, \\
 r_{0,p}(\cdot, x) & \doteq \frac{1}{p!} \int_0^1 (1-s)^p \frac{\ud^p}{\ud s^p} t_0[g^{(s)}, m^{(s)}, A^{(s)}](\cdot, x) \ud s.
\end{align*}
By choosing $p$ large enough, one obtains a distribution $r_{0, p}$ with a low enough scaling degree to have a unique extension that preserves the scaling degree \cite{BrunettiFredenhagenScalingDegree}. Hence, it suffices to extend the $\tau_0$'s. As shown in the following, these may be decomposed as
\[
 \tau_{0, l}^{\underline \alpha}(y,x) = \sum_{\underline a \underline \mu} C_{\underline a \underline \mu}(x) \exp_x^* {u_{0,l}}^{\underline \alpha \underline{a} \underline{\mu}}(y).
\]
Here $C$ is a Lorentz and gauge tensor of mass dimension $l$ built from $g_{\mu \nu}$, and (covariant derivatives) of the curvature, the mass, and the field strength, all evaluated at $x$. The index $\underline{a}$ is a gauge multiindex. The distributions $u_{0,l}$ are spinorial, Lorentz, and gauge tensors, which are Lorentz invariant,
\begin{equation}
\label{eq:uLorentz}
 {u_{0,l}}^{\underline \alpha \underline a \underline \mu}( \Lambda(S) \cdot ) = S^{\underline \alpha}_{\underline \beta} \lambda(S)^{\underline \mu}_{\underline \nu} {u_{0,l}}^{\underline \beta \underline a \underline \nu}( \cdot ),
\end{equation}
where $S \in \Spin_0$ and the action on the $\alpha$ indices is on the spinorial and the tensorial component. They are also gauge invariant in the following sense:
\begin{equation}
\label{eq:uGauge}
 \rho( g )^{\underline \alpha}_{\underline \beta} \rho( g )^{\underline a}_{\underline b}  {u_{0,l}}^{\underline \beta \underline b \underline \mu} = {u_{0,l}}^{\underline \alpha \underline a \underline \mu}.
\end{equation}
Their scaling degree is $q-l$, where $q$ is the scaling degree of $t_0$ at the diagonal.

To prove this, one proceeds as follows: One assumes that $g_{\mu \nu}$, $m$, and $A_\mu$ are polynomials in $\xi$. They are thus entirely determined by the value of their derivatives, i.e., the jet space, at the origin. As a jet space basis of $A_\mu$, we may choose the following:
\begin{equation*}
 \del_{(\mu_1} \dots \del_{\mu_k} A_{\nu)}, \quad \nabla_{(\mu_1} \dots \nabla_{\mu_k} F_{\nu) \lambda}, \quad k \in \N_0.
\end{equation*}
If we now consider the infinitesimal version of \eqref{eq:t_0_gaugeInv}, we see that $t_0$ may not depend on the derivatives of $A$, as otherwise derivatives of the gauge parameter would appear which do not have to vanish at the diagonal. Hence, we may compute the $\tau_{0,l}$ as follows:
\begin{align*}
 \tau_{0,k}(y,x) & = \sum_{k = \sum_j (j l_j + (j+2) m_j + (j+1) p_j)} c_{lmp} \\
 & \times \frac{\del^{\sum (l_j + m_j + p_j)}}{\prod_j \del^{l_j} g_{\mu \nu, \sigma_1 \dots \sigma_j} \del^{m_j} F_{\mu \nu; \sigma_1 \dots \sigma_j} \del^{p_j} m_{,\sigma_1 \dots \sigma_j}} t_0(y,x)|_{g = \eta, m = 0, A = 0} \\
 & \times \prod_j (g_{\mu \nu, \sigma_1 \dots \sigma_j})^{l_j} (F_{\mu \nu; \sigma_1 \dots \sigma_j})^{m_j} (m_{,\sigma_1 \dots \sigma_j})^{p_j}.
\end{align*}
Here $c_{lmp}$ is a combinatorical factor. The factor in the third line then gives the tensor $C$, whereas the factor in the second line gives the distributions $u_{0,k}$. The gauge invariance \eqref{eq:uGauge} is now a consequence of \eqref{eq:t_0_gaugeInv} and the fact that we evaluate at $A = 0$. Note in particular that the distributions $u_{0,k}$ do not depend on the background fields any more, so they are ``universal'', and their extension to the origin defines a coherent extension of $t_0$ on all backgrounds simultaneously. For a discussion of how to extend in a way that preserves Lorentz invariance \eqref{eq:uLorentz} and almost homogeneous scaling, we again refer to \cite{HollandsWaldTO}. The preservation of the gauge symmetry is then straightforward if the gauge group is compact: For an extension $u$ define
\[
 \tilde u^{\underline \alpha \underline a \underline \mu} = \int_G \rho( g )^{\underline \alpha}_{\underline \beta} \rho( g )^{\underline a}_{\underline b} u_{\underline c}^{\underline \beta \underline b \underline \mu} \ud g.
\]
One then proceeds as in \cite{HollandsWaldTO} to arrive at:\footnote{The argument showing that it is possible to redefine the time-ordered products such that \eqref{eq:TO_eom} holds can be found in \cite{HollandsWaldStress}.}
\begin{proposition}
There exist covariant time-ordered products for any compact gauge group.
\end{proposition}

\section{Currents and conservation laws}
\label{sec:Currents}

As discussed in Remark~\ref{rem:WickUniqueness}, the definition of Wick powers, i.e., time-ordered products at first order, is not unique, but allows for some renormalization freedom. To ease the discussion, let us introduce a succinct notation for the Wick powers. For example, the Wick power $\Psi = \psi^+_a \psi^b$, evaluated in a test tensor $t$ and on a configuration $B$, is given by
\[
 \Psi_{(SM, P)}(t)(B) = \int t^a_{\ b}(x,x) (B_2)_a^{\ b}(x,x) \ud_g x,
\]
for $t \in \Gamma_c^\infty(M, D_\rho^\oplus M \wedge D_\rho^\oplus M)$, where we pick the component of $B_2$ whose first entry is in the dual Dirac bundle and whose second entry in the Dirac bundle. Here $a,b$ stand for combined spinor and gauge indices. Derivatives on one of the $\psi$'s translate into derivatives on the corresponding variable of $B_2$.

For example, to modify the definition of the Wick power $\Psi = \psi^+_a \psi^b$, we can modify the parametrix as
\[
 H_a^b(x,y) \to H_a^b(x,y) + \delta H_a^b(s(x,y)),
\]
where $s(x,y)$ is the point $\gamma_{x,y}(1/2)$, where $\gamma_{x,y}: [0,1]: M \to M$ is the unique geodesic from $x$ to $y$, and $\delta H(z)$ is a covariant tensor defined from the jet of the background fields at $z$. Due to the scaling axiom for time-ordered products, only modifications $\delta H$ with the correct scaling dimension are admissible, so for $n=4$, we have possibilities like
\[
 \delta H = \alpha_0 m^3 + \alpha_1 m R + \alpha_2 i m \gamma^\mu \gamma^\nu F_{\mu \nu} + \alpha_4 \gamma^\mu \nabla_\mu R. 
\]
Such a redefinition of course also affects the Wick product $\Psi_\mu = \nabla_\mu \psi^+_a \psi^b$. However, it can also be redefined independently, by
\begin{equation}
\label{eq:deltaH_nabla}
 H_a^b(x,y) \to H_a^b(x,y) + {\delta H_\mu}_a^b(s(x,y)) \del^\mu_x s(x,y),
\end{equation}
where $\delta H_\mu$ is a covariant tensor of mass dimension $4$ (for $n=4$). The ambiguity in the definition of Wick powers was first discussed by Hollands and Wald for the scalar case \cite{HollandsWaldWick}. They also showed that it may be used to achieve a conserved stress-energy tensor for the scalar field in dimension $n>2$ \cite{HollandsWaldStress}. In the following, we perform an analogous analysis for the Dirac field.

\subsection{Current conservation}
\label{sec:CurrentConservation}

We want to show that with our choice of the parametrix, the current
\[
 j^\mu_\alpha = \tr \psi^+ T_\alpha \gamma^\mu \psi
\]
is covariantly conserved. Here we used the same succinct notation as above. $T_\alpha$ is a generator of $\g$ in the representation $\rho$ and the trace is over the spinor and the gauge indices. Note that in order to view this as a field, one has to enlarge the space of test tensors to also include sections of the bundle $P \times_{\ad} \g \otimes T^*M$. For the divergence of the current, we compute
\[
 \nabla_\mu (\tr \psi^+ T_\alpha \gamma^\mu \psi) = \tr (D \psi)^+ T_\alpha \psi - \tr \psi^+ T_\alpha D \psi,
\]
so the $\order(\hbar^0)$ of this Wick power vanishes weakly, i.e., on all on-shell configurations. A possible violation of current conservation can thus only stem from the $\order(\hbar)$ term, i.e., the parametrix. In order to determine this violation, we have to compute
\begin{multline}
\label{eq:H-H}
 H(D^* f', f) - H(f', Df) = \tfrac{1}{4} \left( \tilde h_+(P^* f', f) - \tilde h_+(f', P f) \right. \\
 \left. + \overline{\tilde h_-(f^+, P {f'}^+)} - \overline{\tilde h_-(P^* f^+, {f'}^+)} \right),
\end{multline}
or more precisely, determine the coinciding point limit of the corresponding distribution, and trace it with generators of the gauge group. The computations performed in \cite[Lemma~2.1]{MorettiStressEnergy} also apply to the present case, so
\[
 [\tilde h_\pm \circ P](x) = \begin{cases} 0 & n \text{ odd} \\ c_n [V_{n/2}](x) & n \text{ even}, \end{cases}
\]
where $V_i$ are the Hadamard coefficient for $P$, $c_n$ are real numbers, and the square brackets denote the coinciding point limit.
%Note again that $h_\pm$ and $\tilde h_\pm$ differ by a smooth function all of whose derivatives vanish on the diagonal.
Furthermore, we have, by \eqref{eq:h*}, $\tilde h_+(P^* f', f) = \tilde h^*_-(f, P^* f')$, so that
\[
 [P \circ \tilde h_\pm](x) = \begin{cases} 0 & n \text{ odd} \\ c_n [V_{n/2}^*]^*(x) & n \text{ even}, \end{cases}
\]
where $V_i^*$ are the Hadamard coefficients for $P^*$. As a consequence of \cite[Thm.~6.4.1]{Friedlander}, we have $[V_n^*]^* = [V_n]$, so the contributions in \eqref{eq:H-H} cancel. So we have shown:
\begin{proposition}
With the parametrix as defined in Section~\ref{sec:Fields}, the current is covariantly conserved.
\end{proposition}

\begin{remark}
The proof relies on the fact that we may use the combination $\tilde h_+(D^* f', f) + \tilde h_+(f', D f)$ in the definition of the parametrix $H$. This is not possible if $D$ is not an endomorphism, as in the case of chiral fermions. Hence, one expects the usual chiral anomalies. The relevance of $D$ being an endomorphism for the occurrence of anomalies was already discussed in \cite{AlvarezGaumeWitten} from the point of view of the Euclidean path integral.
\end{remark}

Let us discuss the remaining ambiguity. A redefinition leading to
\[
 j^\mu_\alpha \to j^\mu_\alpha + r^\mu_\alpha
\]
would require the existence of a local and covariant vector $r^\mu_\alpha$ that is conserved. The only such vector is the external current $J^\mu_\alpha$ responsible for the background field. It follows that $j^\mu_\alpha$ is uniquely defined up to multiples of $J^\mu_\alpha$. In particular it is unique in regions that are void of charges and currents. The fact that for quantum electrodynamics in external potentials, there is an ambiguity proportional to the external current was already discussed by Schwinger \cite{SchwingerQED2}, in a setting where the external potential was treated as a perturbation. This can be interpreted as a charge renormalization. Evaluation in a state (which amounts to computing a certain limit of the difference of the corresponding two-point function and the parametrix) then yields the expectation value of the current, which could be used to estimate back-reaction effects.

\begin{remark}
For the case of a flat background and the gauge group $G = U(1)$ in the fundamental representation, the use of a local renormalization scheme based on the parametrix was already proposed by Marecki \cite[Sec.~VI.7]{MareckiThesis}.\footnote{In the seminal work of Euler and Heisenberg \cite{EulerHeisenberg}, one finds the same approach of subtracting a (essentially unique) reference object from the two-point function (which involves the choice of a state). This idea goes back to Dirac \cite{Dirac34}. However, the reference object Euler and Heisenberg employ is not the parametrix. The difference is not only that there they use the mass to fix the scale $\Lambda$, but also the Hadamard coefficients $V_k$ disagree. For example, the coinciding point limit of the analog of $V_1$ vanishes in \cite{EulerHeisenberg}, in contrast to the parametrix, \cf \eqref{eq:V_1}. This seems to stem from \cite{Heisenberg34}, where, for some unknown reason, only terms at most linear in the $\gamma$-matrices are considered.} However, the discussion of the ambiguities given there is not completely satisfactory, as the need for a covariant prescription seems not to be fully taken into account. Other definitions of the renormalized current one finds in the literature usually rely on the existence of a ground state, i.e., they require an ultrastatic spacetime with time-independent background fields \cite{DoschMuller}. The usual method to compute back-reaction effects is via effective actions, \cf \cite{DittrichReuter, Dunne04} for an overview. In this approach, the dependence on the state seems obscure.
We plan to address the issue of consequences of the local definition in a forthcoming joint work with M.~Wrochna.
\end{remark}

\subsection{The stress-energy tensor}
\label{sec:StressEnergy}

The renormalization freedom of Wick powers was used by Hollands and Wald to construct a conserved stress-energy tensor in the scalar case \cite{HollandsWaldStress}. Here, we perform the analogous analysis for the case of charged Dirac fields.

The first thing to notice is that the stress-energy tensor is in general only conserved if all fields are on-shell. Unless we are given a Lagrangean for the Yukawa background field $m$, variation \wrt $m$ leads to $\tr \psi^+ \psi = 0$. We thus have two choices: Either we assume that background fields are absent, with the possible exception of a constant mass (which does not lead to problems with the stress-energy tensor). Or we assume that the background fields are equipped with some Lagrangean. But then the coupling to the Dirac fields should be treated perturbatively, as otherwise terms involving the Wick square $\psi^+ \psi$ would enter the equation of motion of the background fields. But this reduces us to the first case for the free theory.

Hence, let us consider a charged Dirac field with a possibly non-zero mass $m$ in a gauge field background with vanishing curvature. The stress-energy tensor for this field is given by \cite{ForgerRomer}
\begin{multline}
\label{eq:T_munu}
 T_{\mu \nu} = \tr \Big[ \tfrac{1}{2} \left( \nabla_{(\mu} \psi^+ \gamma_{\nu)} \psi - \psi^+ \gamma_{(\mu} \nabla_{\nu)} \psi \right) \\ - \tfrac{1}{2} g_{\mu \nu} \left( \nabla_\lambda \psi^+ \gamma^\lambda \psi - \psi^+ \gamma^\lambda \nabla_\lambda \psi + 2 m \psi^+ \psi \right) \Big],
\end{multline}
where the trace is over gauge and spinor indices.
In terms of the Wick squares defined above,
this may be written as
\begin{multline*}
 T_{\mu \nu} = \tr \Big[ \tfrac{1}{2} \left( \gamma_\nu \Psi_\mu + \gamma_\mu \Psi_\nu \right) - \tfrac{1}{4} \left( \gamma_\mu \nabla_\nu \Psi + \gamma_\nu \nabla_\mu \Psi \right) \\
 - g_{\mu \nu} \left( \gamma^\lambda \Psi_\lambda - \tfrac{1}{2} \gamma^\lambda \nabla_\lambda \Psi + m \Psi \right) \Big].
\end{multline*}
For its divergence and trace, we obtain
\begin{align*}
 \nabla^\mu T_{\mu \nu} 
 & =  \tr \left[ \tfrac{1}{2} \left( \gamma_\nu \nabla^\mu \Psi_\mu + \gamma^\mu \nabla_\mu \Psi_\nu - 2 \gamma^\mu \nabla_\nu \Psi_\mu \right) \right. \\
 & \quad \left. - \tfrac{1}{4} \left( \gamma_\nu \nabla^\mu \nabla_\mu \Psi - \gamma^\mu \nabla_\nu \nabla_\mu \Psi + \gamma^\mu R_{\mu \nu} \Psi \right) - m \nabla_\nu \Psi \right], \\
 g^{\mu \nu} T_{\mu \nu} & = \tr \left[ (1-n) \gamma^\mu \Psi_\mu - \tfrac{1}{2}(1-n) \gamma^\mu \nabla_\mu \Psi - n m \Psi \right].
\end{align*}
Here we used
\[
 [\nabla_\mu, \nabla_\nu] \Psi = \mathfrak{R}_{\mu \nu} \Psi - \Psi \mathfrak{R}_{\mu \nu},
\]
where $\mathfrak{R}$ is the spin curvature tensor, which fulfills \cite{HackThesis}
\[
 \mathfrak{R}_{a b} \gamma^b = - \gamma^b \mathfrak{R}_{a b} = \tfrac{1}{2} R_{a b} \gamma^b.
\]

As for the divergence of the current, the divergence of the stress-energy tensor is a c-number modulo a weakly vanishing functional. Let us assume that this c-number is of the form $\nabla^\mu Q_{\mu \nu}$, for $Q_{\mu \nu}$ symmetric and locally and covariantly constructed. For $n=4$ this follows from the results of \cite{DHP09}. For the generic case, it was conjectured in \cite{HollandsWaldStress} that this is the case for all parity preserving models\footnote{There are counterexamples in parity violating theories, \cf \cite{AlvarezGaumeWitten}. The fact that these are all possible purely gravitational anomalies \cite{BrandtDragonKreuzer90} suggests that this is indeed fulfilled.}. To achieve a conserved stress-energy tensor, one may then use the redefinition \eqref{eq:deltaH_nabla} of the parametrix to modify
\[
 \Psi_\mu \to \Psi_\mu - N^{-1} 2^{-[n/2]} \left( \gamma^\nu Q_{\mu \nu} - \tfrac{1}{n-1} \gamma_\mu Q^\lambda_\lambda \right),
\]
where $N$ is the dimension of the gauge representation. Note that such a redefinition does not affect the current, so both current and stress-energy conservation can be achieved. Also note that there are no restrictions on the dimension $n$, in contrast to the scalar case \cite{HollandsWaldStress}, where one has $n-2$ in the denominator, so that one can achieve conservation only for $n>2$.
%Hence, we have proven:
%\begin{proposition}
%There are no algebraic obstructions to achieving a conserved stress energy tensor, in arbitrary dimension $n$.
%%The renormalization freedom of Wick powers allows to achieve a conserved stress energy tensor, in arbitrary dimension $n$.
%\end{proposition}
Hence, we have shown that there are no algebraic obstructions to achieving a conserved stress energy tensor, in arbitrary dimension.
If the above assumption is valid, as for $n=4$, this implies that the Wick powers may indeed be modified such that the stress-energy tensor is conserved in any dimension.

\begin{remark}
There is another prescription for obtaining a conserved stress-energy tensor, due to Moretti \cite{MorettiStressEnergy}. There, one directly changes the stress-energy tensor by adding a Wick monomial that vanishes on-shell. In the scalar theory, one uses
\[
 T'_{ab} = T_{ab} + c g_{ab} \varphi P \varphi,
\]
where $P$ is the wave operator. In the case of the Dirac field, this was adapted as \cite{DHP09}
\[
 T'_{ab} = T_{ab} + c g_{ab} \psi^+ D \psi.
\]
While the two methods give the same expectation values of the stress energy tensor, and thus are equivalent for the purpose of discussing the semi-classical Einstein equation, there are important conceptual differences. As noted in \cite{HollandsWaldStress}, it seems highly unlikely that Moretti's description can be generalized to the interacting case, in contrast to the method of Hollands and Wald. In particular, the redefinition of the Wick powers such that the stress energy tensor is conserved is the first step in constructing time-ordered products that fulfill the principle of perturbative agreement\footnote{It states that the physics is independent of the choice of a background, i.e., of the split of the action into a free and an interacting part (provided the free part is at most quadratic in the fields).} of \cite{HollandsWaldStress}. Such a choice of time-ordered products will automatically ensure the conservation of the stress-energy tensor also in the interacting case. The fact that in the two-dimensional scalar case a conserved stress-energy tensor can not be achieved by a redefinition of Wick powers, whereas no such restriction exists for Moretti's description, further shows that the two methods are not equivalent.
\end{remark}

Let us close this section by discussing the remaining renormalization freedom for $n=4$. After achieving a conserved stress-energy tensor, the remaining freedom must preserve this conservation. Hence, it may only be modified as
\begin{equation}
\label{eq:T_Ambiguity}
 T_{\mu \nu} \to T_{\mu \nu} + \beta_0 I_{\mu \nu} + \beta_1 J_{\mu \nu}, + \beta_2 m^2 G_{\mu \nu} + \beta_3 m^4,
\end{equation}
where $I_{\mu \nu}$ and $J_{\mu \nu}$ are the two linearly independent conserved curvature tensors of dimension $4$ (obtained by variation \wrt $g^{\mu \nu}$ of $R^2$ and $R_{\mu \nu} R^{\mu \nu}$). Such changes may indeed be achieved, by the redefinition (again performed via the redefinition \eqref{eq:deltaH_nabla} of the parametrix)
\[
 \Psi_\mu \to \Psi_\mu + \tfrac{1}{4} \gamma^\nu \left( \delta T_{\mu \nu} - \tfrac{1}{3} g_{\mu \nu} \delta T^\lambda_\lambda \right).
\]
Hence, one has the same renormalization ambiguities of the stress-energy tensor as for scalar fields, as conjectured in \cite{DHP09}.
In particular, these suffice to cancel the term $\Box R$ in the trace anomaly \cite{DHP09} (as $I_{\mu \nu}$ and $J_{\mu \nu}$ have trace proportional to $\Box R$).
However, let us note that if we treat the Yukawa background field completely perturbatively, then $m = 0$ and the last two terms in \eqref{eq:T_Ambiguity} are absent. These are replaced by one new ambiguity, namely the stress-energy tensor of the background Yukawa field, at zeroth order in perturbation theory. Similarly, one may add a multiple of the stress-energy tensor of the gauge background field.

\subsection{Scaling behavior}
\label{sec:Scaling}
We briefly comment on how a variant of the background field method \cite{Honerkamp72} can be used to determine the scaling behavior or renormalization group flow at $\order(\hbar)$. For even dimension $n$, one has a non-trivial scaling behavior of Wick powers, since, as discussed in Remark~\ref{rem:Scaling}, the parametrix involves a logarithmic term, which necessitates the choice of a scale $\Lambda$. But due to local covariance, this choice must be done simultaneously on all backgrounds. Hence, for a Wick square $\Psi$, we in general have
\begin{equation}
\label{eq:NontrivialScaling}
 S_\lambda \Psi = \lambda^{d_\Psi} \Psi + \hbar r \log \lambda,
\end{equation}
where $r$ is a local covariant object, and $d_\Psi$ is the scaling dimension of $\Psi$. In order to interpret this result, consider the backgrounds $(M, g, A, m)$ and $(M, \lambda^2 g, A, \lambda^{-1} m)$ as described above \eqref{eq:Def_sigma}. The choice of a definition of a Wick square $\Psi$ should correspond to the design of a corresponding measurement apparatus. This apparatus involves a linear length $L$, which by definition is the same on all backgrounds. Now the conformal map $(M, g) \to (M, \lambda^2 g)$ maps the apparatus to one of length $\lambda L$. Hence, comparing $\Psi$ and $S_\lambda \Psi$ amounts to comparing two definitions of $\Psi$ related by a different choice of a length scale. This is obviously in close analogy to the comparison of field theories defined at different renormalization scales, which is the idea underlying the Callan--Symanzik equation. The difference is that in the present setting, it already applies to Wick powers. We refer to \cite{HollandsWaldRG, BDF09} for a deeper discussion of the connection of scaling to the usual notions of the renormalization group flow.

As noticed in the preceding subsection, ultimately the background fields should be determined dynamically, i.e., they should be given some Lagrangean, which, for the sake of simplicity, we assume to be free. The coupling to the Dirac fermions is now an interaction term. Hence, we split the Yukawa and the gauge field into a free and an interacting part, indicated by subscripts $0$ and $1$, respectively:
\begin{align*}
 m & = m_0 + m_1, & A^\mu & = A^\mu_0 + A^\mu_1. 
\end{align*}
The fields $m_1$ and $A^\mu_1$ will be quantized.
We can split the Lagrangean into a free part $L_0$ (involving $m_1$ and $A^\mu_1$ at most quadratically), and an interaction part, given by
\begin{equation}
\label{eq:L_1}
 L_1 = m_1 \psi^+ \psi + i A_{1 \mu} \psi^+ \gamma^\mu \psi.
\end{equation}
In $L_0$, no coupling of $m_1$ or $A_1$ to the Dirac fermion is present, so in particular the parametrix will not contain couplings between these field. Hence, as the fields $m_1$ and $A_{1 \mu}$ enter linearly in \eqref{eq:L_1}, the anomalous scaling of this expression is completely determined by that of the Wick squares $\psi^+ \psi$ and $\psi^+ \gamma^\mu \psi$ (where a trace is understood).

In four spacetime dimensions, the auxiliary Hadamard parametrix is formally given by
\begin{equation*}
 h^\pm_\Lambda(x,x') = \frac{1}{16 \pi^2} \lim_{\varepsilon \to \pm 0} \left( 4 \frac{V_0(x,x')}{\Gamma_\varepsilon(x,x')} + \log \frac{-\Gamma_\varepsilon(x,x')}{\Lambda^2} V(x,x') \right),
\end{equation*}
where
\begin{equation}
\label{eq:V_Expansion}
 V = \sum_{k=0}^\infty \frac{1}{2^{2k}(k+1)! k!} \Gamma^{k} V_{k+1} .
\end{equation}
Noting that the Hadamard parametrix is obtained by applying $D$, we see that in order to compute the scaling behavior of the above expressions, we have to know the coinciding point limit of $V_1$ up to the first order derivative. For these, we obtain, for the case of electrodynamics ($G = U(1)$ and the fundamental representation)
\begin{align}
\label{eq:V_1}
 [V_1] & = - \tfrac{1}{12} R - \tfrac{i}{4} [\gamma^\lambda, \gamma^\rho] F_{\lambda \rho} - (n-1) m^2, \\
 [\tilde \nabla_\mu V_1] & = - \tfrac{1}{24} \nabla_\mu R - \tfrac{i}{8} [\gamma^\lambda, \gamma^\rho] \nabla_\mu F_{\lambda \rho} - (n-1) m \del_\mu m \nonumber \\
 & \quad - \tfrac{1}{6} \tilde \nabla^\lambda(\mathfrak{R}_{\mu \lambda} - i F_{\mu \lambda} - \gamma_\lambda \del_\mu m + \gamma_\mu \del_\lambda m + m^2 [\gamma^\mu, \gamma^\nu] ), \nonumber
\end{align}
%\begin{align}
%\label{eq:V_1}
% [V_1] & = - \tfrac{1}{12} R - \tfrac{i}{4} [\gamma^\lambda, \gamma^\rho] F_{\lambda \rho} - m^2 + \gamma^\lambda \del_\lambda m, \\
% [\nabla_\mu V_1] & = - \tfrac{1}{24} \nabla_\mu R - \tfrac{i}{8} [\gamma^\lambda, \gamma^\rho] \nabla_\mu F_{\lambda \rho} - m \del_\mu m + \tfrac{1}{2} \gamma^\lambda \nabla_\mu \del_\lambda m \nonumber \\
% & \quad - \tfrac{1}{6} \nabla^\lambda(\mathfrak{R}_{\mu \lambda} - i F_{\mu \lambda}), \nonumber
%\end{align}
where $\mathfrak{R}$ is the spin curvature and the square brackets denote the coinciding point limit. For simplicity, we used $m$, $A$ instead of $m_0$, $A_0$. The coefficient $r$ of non-trivial scaling in \eqref{eq:NontrivialScaling} is now proportional to (here one uses $[\mathfrak{R}_{\mu \nu}, \gamma_\lambda] = R_{\mu \nu \rho \lambda} \gamma^\rho$, \cf \cite[Lemma~I.2.2.9]{HackThesis})
\[
 r \sim 2 m_1 \nabla^\lambda \nabla_\lambda m_0 - \tfrac{1}{3} R m_0 m_1 - 4 m_1 m_0^3 + \tfrac{4}{3} A_1^\mu \nabla^\lambda F_{0, \lambda \mu}.
\]
This is, up to total derivatives, the expansion to linear order in $m_1$, $A^\mu_1$ of
\[
 - \nabla^\lambda m \nabla_\lambda m - \tfrac{1}{6} R m^2 - m^4 - \tfrac{1}{3} F^{\mu \nu} F_{\mu \nu},
\]
which is the Lagrangean for a conformally coupled scalar $m^4$ theory and a Yang--Mills Lagrangean. Up to the purely gravitational terms (which we can not obtain here, unless we also split the metric), this coincides with the $a_4$ term of the bosonic part of the spectral action of Chamseddine and Connes \cite{ChamseddineConnesSA}. In the present setting we obtain it from the fermionic part through scale transformations, on generic globally hyperbolic spacetimes (in contrast to the compact Riemannian spaces needed for the spectral action). We note that working on compact Riemannian spaces and using a cut-off, all terms of the bosonic part of the spectral action can be obtained by scale transformations \cite{AndrianovLizziKurkov11}.

In a forthcoming publication, we will examine this further, showing that if a generalization of the principle of perturbative agreement \cite{HollandsWaldStress} to the case of gauge backgrounds holds, then the fermionic contribution to renormalization group flow at the one-loop level can indeed be calculated as sketched above. Noting that the coinciding point limits of the Hadamard coefficients are related to the (Euclidean) Seeley-deWitt coefficients appearing in the heat kernel expansion, this establishes a connection to the heat kernel method. See \cite{HackMoretti} for related discussions.

%It would certainly be desirable to better understand this, i.e., to establish a proof that the above method can really be used to calculate the renormalization group flow of the interacting model at the one loop level. It seems that this would be a prerequisite for the fulfillment of the principle of perturbative agreement introduced by Hollands and Wald \cite{HollandsWaldStress}. 

\subsection*{Acknowledgments}
I would like to thank Dorothea Bahns, Thomas-Paul Hack, Harold Steinacker, and especially Micha{\l} Wrochna for helpful discussions and remarks.
I am very grateful to Kartik Prabhu for communicating a mistake in an earlier version of the manuscript and for discussions on this point.
I am also grateful to an anonymous referee for a very careful reading of the manuscript.
A large part of this work was carried out at the Courant Research Centre ``Higher Order Structures'' at the University of G\"ottingen.
This work was supported by the German Research Foundation (Deutsche Forschungsgemeinschaft (DFG)) through the Institutional Strategy of the University of G\"ottingen and by the Austrian Science Fund (FWF) under the contract P24713.

\appendix

\section{The $\Spin$ group}
\label{app:Spin}

We recall some basic material on the $\Spin$ group, \cf \cite{LawsonMichelsohn,StrohmaierSemiRiemannian} for more details. We denote by $\Cl(n)$ the real Clifford algebra corresponding to the bilinear form $\eta$ with signature $(-, +, \dots, +)$ on $\R^n$, i.e., the algebra generated by the identity $\1$ and elements $e^\mu$ subject to
\[
 \{ e^\mu, e^\nu \} = 2 \eta^{\mu \nu} \1.
\]
By defining the involution ${e^\mu}^* = - e^\mu$, one obtains the complexified Clifford algebra $\Cl^c(n)$. There is an algebra isomorphism from $\Cl^c(n)$ to $\Mat_\C(2^{[n/2]})$ for even $n$ and to $\Mat_\C(2^{[n/2]}) \oplus \Mat_\C(2^{[n/2]})$ for odd $n$ (note that this is not yet a $*$ isomorphism), \cf \cite{StrohmaierSemiRiemannian} for a concrete realization. To obtain an irreducible representation, one restricts to the first summand for odd $n$. One equips $\C^{2^{[n/2]}}$ with the inner product
\[
 (v, w) = - i \skal{v}{\gamma^0 w}_{\C^{2^{[n/2]}}},
\]
where $\gamma^\mu$ is the image of $e^\mu$ in this representation and $\skal{\cdot}{\cdot}_{\C^{2^{[n/2]}}}$ is the standard inner product on $\C^{2^{[n/2]}}$. With this inner product, $\C^{2^{[n/2]}}$ is a Krein space and the representation is a $*$ representation (mapping the involution to the Krein-adjoint). The inner product is invariant under the identity component $\Spin_0(n)$ of the $\Spin(n)$ group, defined as
\[
 \Spin(n) \doteq \{ s \in \Cl(n)| s = u_1 \dots u_{2k}, u_i \in \R^n, u_i^2 = \pm 1 \},
\]
where we identified $\R^n$ with a subspace of $\Cl(n)$ via $v_\mu \tilde e^\mu \to v_\mu e^\mu$, with $\{ \tilde e^\mu \}$ an orthonormal basis of $\R^n$.
%Its identity component is denoted by $\Spin_0(n)$.
There is a canonical homomorphism from $\Spin_0(n)$ to the identity component $\Lor_0$ of the Lorentz group $SO(n-1, 1)$. For $n>2$, this is a double covering, whereas for $n=2$, both groups are isomorphic to $\R$. The restriction of the above irreducible representation of $\Cl^c(n)$ to $\Spin_0(n)$ is the \emph{spinor representation}. It is irreducible for odd $n$ and reducible for even $n$, decomposing into two irreducible chiral representations.

\section{Deformation of the background}
\label{app:Deformation}

\begin{proof}[Proof of Proposition~\ref{prop:Deformation}]
By \cite{BernalSanchez03, BernalSanchez05}, $M$ is diffeomorphic to $\R \times \Sigma$, with a smooth time function $t$. We define $M', \tilde M = M$ as smooth oriented manifolds. Similarly, we define $P', \tilde P = P$ and $SM', S\tilde M = SM$ as smooth principal bundles. With the induced metric $h$, $\Sigma$ is a Riemannian manifold, and there exists a Riemannian metric $\tilde h$, conformal to $h$, such that $(\Sigma, \tilde h)$ is complete \cite{NomizuOzeki}. We define $\tilde M = \R \times \Sigma$ with the metric $\tilde g = - \ud t^2 \otimes \tilde h$. By \cite[Thm.~2.54]{BeemEhrlich}, it is globally hyperbolic. Now one proceeds as in \cite[Prop.~C.1]{FNW81} to define a metric on $M'$ that interpolates between $\Sigma$ (at $t=0$) and $\tilde \Sigma$ (at $t=-1$). Both $M'$ and $\tilde M$ inherit the time-orientation from $M$ (by the orientation of $\del_t$).

Regarding the spin structure, we note that for $n>2$, spin structures are classified (up to equivalence) by $H^1(M ; \Z_2)$, i.e., by assigning a sign to each nontrivial cycle, indicating whether in the covering of the frame bundle by the spin bundle one changes the sheet when following the cycle \cite[Thm.~1.7]{LawsonMichelsohn}.
This is purely topological, so by choosing the same assignment as for $(SM,\pi_S)$, we define the spin structures $(SM', \pi'_S)$ and $(S\tilde M, \tilde \pi_S)$.\footnote{For a concrete prescription of how to change the spin projection under deformations of the metric, \cf the discussion preceding the introduction of the scaling transformation in Section~\ref{sec:Fields}.}
For $n=2$, $\Spin_0$ and $\Lor_0$ are isomorphic, so the spin structure is unique (up to equivalence).
%Extending a section $s_i$ of $S \tilde M|_{\tilde U_i}$ by parallel transport along the time direction to all of $U_i$, one obtains the section $\tilde s_i$.
%Extending a section $s$ of $S\tilde M|_\Sigma$ by parallel transport along the time direction to all of $\tilde M$, one obtains the section $\tilde s$.

%It remains to discuss the connections $\tilde A$, $A'$.
To construct $\tilde A$, take some connection $\hat A$ on $\tilde P|_{\tilde \Sigma}$ (existence is guaranteed by \cite[Thm.~II.2.1]{KobayashiNomizu}). Choose an open cover $\{ U_i \}$ of $\tilde \Sigma$, where each $U_i$ is topologically trivial, and corresponding local sections $\hat s_i$. In a pull-back \wrt these, the connection $\hat A$ is of the form $\hat s_i^* \hat A(x) = \hat A_{i, a}(x) \ud x^a$, where $x^a$ are local coordinates on $U_i$. Choose some equivariant lift $v$ of $\del_t$ to $\tilde P$ and extend the sections $\hat s_i$ to sections $s_i: \R \times U_i \to \tilde P$ by taking the integral curves $c_i(x, t)$ of $\hat s_i(x)$ \wrt $v$ and defining $s_i(t, x) = c_i(x,t)$. By construction, $v|_{s_i(t,x)} = s_{i *} \del_t|_{(t,x)}$. Then define the connection $\tilde A$ by its pull-backs $s_i^* \tilde A(t,x) = \hat A_{i, a}(x) \ud x^a$. By \cite[Prop.~II.1.4]{KobayashiNomizu}, this defines a connection on $\tilde P$, as $\hat A$ is a connection and the transition functions $\psi_{ij}: \R \times (U_i \cap U_j) \to G$ corresponding to the sections $s_i$ are time-independent, by construction. Furthermore, the horizontal lift of $\del_t$ \wrt $\tilde A$ is $v$ (by equivariance, it suffices to show that $\tilde A(v|_{s_i(t,x)}) = 0$, for all $t, x$ which follows from $v|_{s_i(t,x)}$ being the push-forward of $\del_t|_{(t,x)}$ along $s_i$ and the definition of $s_i^* \tilde A$). It remains to show that $(\Lie_v \tilde A)(w) = 0$ for all $w \in T \tilde P$. By equivariance, it suffices to consider $w \in T \tilde P|_{s_i(t,x)}$. If $w$ is vertical, the equality follows from standard arguments, in particular that the Lie bracket of a vertical and a horizontal vector field is horizontal. If $w$ is horizontal, we can decompose it into a vertical vector and a vector which is the push-forward of a vector $u \in T\tilde M$ along $s_i$. For this component, we have
\[
 (\Lie_v \tilde A)(s_{i *} u) = (s_i^* \Lie_v \tilde A)(u) = (\Lie_{\del_t} s_i^* \tilde A)(u),
\]
which vanishes by the definition of $s_i^* \tilde A$. In the last step, we used $s_i \circ \phi^\tau_{\del_t} = \phi^\tau_v \circ s_i$, where $\tau \mapsto \phi^\tau_v$ is the one-parameter family of diffeomorphisms generated by $v$.

The interpolating connection $A'$ can now be defined as
\[
 A'(p) =  f(t(p)) i^*_P A + (1-f(t(p))) \tilde \imath^*_P \tilde A,
\]
%\[
% s_i^* A'(t,x) =  f(t) s_i^* A(t,x) + (1-f(t)) s_i^* \tilde A_i(t,x),
%\]
where $f \in C^\infty(\R, [0,1])$ and $f(t) = 1$ for $t > -1/4$ and $f(t) = 0$ for $t < - 3/4$. Here $i_P$ and $\tilde \imath_P$ are the bundle isomorphisms $i_P: P' \to P$, $\tilde \imath_P: P' \to \tilde P$.
For the Yukawa field $m'$, one proceeds in the obvious way.
\end{proof}

%\bibliography{../../mybib}{}
%%%%\bibliographystyle{amsalpha}
%\bibliographystyle{../../h-elsevier_new}

\end{document}